\documentclass[a4paper, 11pt]{article}
\usepackage[english]{babel}
\makeatletter

\usepackage{geometry,graphicx,color}
\usepackage{amsfonts, amsmath, amssymb, slashed,dsfont}
\usepackage{tensor}
\usepackage[pdftex]{hyperref}
\usepackage{epsfig}
\usepackage{pifont}
\usepackage{stmaryrd}
\usepackage{enumitem}
\usepackage{bm}
\usepackage{mathrsfs} 
\usepackage[hyperref,thmmarks,amsmath]{ntheorem}
\usepackage{tikz}
\usetikzlibrary{positioning}
\usetikzlibrary{decorations.markings}
\usetikzlibrary{arrows.meta}

\numberwithin{equation}{section}
\theoremsymbol{}
\theorembodyfont{\slshape}
\theoremheaderfont{\normalfont\bfseries}
\theoremseparator{}
\newtheorem{Theorem}{Theorem}[section]

\newtheorem{proposition}[Theorem]{Proposition}

\theorembodyfont{\upshape}
\theoremsymbol{\ensuremath{\blacklozenge}}
\newtheorem{Example}[Theorem]{Example}

\newtheorem{Remark}[Theorem]{Remark}

\newtheorem{definition}[Theorem]{Definition}
\theoremstyle{nonumberplain}
\theoremheaderfont{\scshape}
\theorembodyfont{\normalfont}
\theoremsymbol{\ensuremath{\blacksquare}}

\newtheorem{proof}{Proof}
\qedsymbol{\ensuremath{_\blacksquare}}
\theoremclass{LaTeX}

\DeclareMathOperator{\actr}{\triangleleft}
\DeclareMathOperator{\actl}{\triangleright}
\DeclareMathOperator{\coactr}{\blacktriangleleft}

\DeclareMathOperator{\bicros}{\actl\!\!\!\blacktriangleleft}

\newcommand{\sign}{\mathrm{sign}}

\DeclareMathOperator{\Hom}{\mathrm{Hom}}

\newcommand{\Moyal}{\mathbb{R}^{2d}_\theta}
\newcommand{\Mink}{\mathbb{R}^{1,d}}
\newcommand{\kM}{\mathbb{R}^{1,d}_\kappa}
\newcommand{\kP}{\mathcal{P}_\kappa}

\newcommand\Der{\mathrm{Der}}
\newcommand\dd{\mathrm{d}}
\newcommand\id{\mathrm{id}}

\newcommand{\tops}[2]{\texorpdfstring{#1}{#2}}

\numberwithin{equation}{section}
\newcommand{\tdl}[2]{\mathrm{d}^{#1}{#2}\;}
\newcommand{\omitel}[1]{\overset{#1}{\vee}}

\DeclareMathOperator{\Lin}{\textnormal{Lin}}
\DeclareMathOperator{\Supp}{\textnormal{Supp}}
\DeclareMathOperator{\Ker}{\textnormal{Ker}}

\newcommand{\institute}[1]{\newcommand{\@institute}{#1}}
\renewcommand{\maketitle}{
{
\center\LARGE\noindent\@title\par
}%
\vspace{1.5\baselineskip}
{
\center\normalsize\noindent\ignorespaces\@author\par
}%
\vspace{0.5\baselineskip}
{
\center\normalsize\ignorespaces\@institute\par
}%
\vspace{2\baselineskip}
}%

\begin{document}

\title{$\kappa$-Minkowski as tangent space I : quantum partition of unity}
\author{Kilian Hersent${}^{a,b}$, Jean-Christophe Wallet${}^{a}$}
\institute{%
\textit{${}^a$ IJCLab, Universit\'e Paris-Saclay, CNRS/IN2P3, 91405 Orsay, France}\\%
\textit{${}^b$ Departamento de Física, Universidad de Burgos, 09001 Burgos, Spain}\\
\bigskip
e-mail:  
\href{mailto:khersent@ubu.es}{\texttt{khersent@ubu.es}}

\href{mailto:jean-christophe.wallet@universite-paris-saclay.fr}{\texttt{jean-christophe.wallet@universite-paris-saclay.fr}}
}%
\maketitle

\begin{abstract}
    We define a quantum (noncommutative) analogue of locally trivial tangent bundle based on two main elements: the definition of local algebras through quotients of ideals of the global algebra as introduced in \cite{Calow_1999}, and the triviality of the local tangent space as being the $\kappa$-Minkowski space inspired from \cite{Brzezinski_1993}. This tangent bundle is explicitly constructed via local coordinate charts. Every local objects are exported to the global algebra through the notion of quantum (noncommutative) partition of unity introduced in this purpose. This partition is also used to export consistently an integral on $\kappa$-Minkowski to an integral on the global algebra.
\end{abstract}
\vspace{20pt}
\tableofcontents
\newpage

\section{Introduction}
\label{sec:intro}
\paragraph{}
Quantum gravity is one of the fundamental question of theoretical physics aiming to properly characterise the behaviour of the gravitational interaction at ultra short distances and very high energy. Recent studies focusing on extreme astrophysical and cosmological events and exploring
the deep quantum regime seem to provide a promising way toward possible observational tests. For a recent review, see \cite{Addazi_2022}.

Various theories and approaches to this challenging question have been considered in the literature so far. From these often emerges a consensus that quantum gravity should lead at an effective regime to a quantum spacetime. Such an object can be suitably described within the framework of noncommutative geometry.

\paragraph{}
This paper is the first of a series aiming to construct a theory of gravity on quantum spacetime in which the properties of this latter, which can be called "quantum base space", are rigidly linked to the properties of a quantum (analogue of) coordinate space, assumed in the following to be the $\kappa$-Minkowski space. This thus mimics in some sense the commutative situation. The idea underlying this construction can be viewed as an extension of the usual situation of commutative manifolds. Indeed, the Minkowski space $\Mink$ arises in coordinate charts $\varphi^\alpha: U_\alpha \to \Mink$, where $U_\alpha$ is a local open set of the full spacetime manifold $\mathcal{M}$. The charts allow to export properties from $\Mink$ to the manifold. In such a formalism, the tangent space of $U_\alpha$, noted $TU_\alpha$, can be shown to be trivial, that is 
\begin{align}
    TU_\alpha \simeq U_\alpha \times \Mink.
    \label{eq:int_loc_tan}
\end{align}
This identity is the one giving rise to the tetrad formalism. Using this correspondence, the properties of $U_\alpha$ can be related to $TU_\alpha$ and thus to the Minkowski space. Still, this statement only involves a local part of the full spacetime. One way to export the objects of $U_\alpha$ to the full $\mathcal{M}$ is to use a partition of unity $\chi$. This last part is usually not exploited in the physics literature, but as exemplified below it is necessary in the coherent definition of global objects, such as an integral over $\mathcal{M}$.

\paragraph{}
In this paper, we adapt the above formulation to a noncommutative framework. The base space $\mathcal{M}$ is described by its algebra of smooth functions $\mathcal{C^\infty(M)}$. This commutative algebra is traded for a noncommutative algebra $\mathcal{A}$, possibly referred to as the base space algebra. As $\mathcal{M}$ is dynamical in general relativity, the algebra $\mathcal{A}$ is also considered dynamical. For now, this only means that it is not fixed {\it{ab initio}}. In order to mimic the local structure of $\mathcal{M}$ as a patch of open sets $\{U_\alpha\}_\alpha$, the quantum base space $\mathcal{A}$ is equipped with a covering of algebra $\{\mathcal{A}_\alpha\}_\alpha$, as done in \cite{Calow_1999}. Here, the local
\footnote{
The notion of locality here (which we refrain to call quantum locality) is a mere extension of the commutative one where the manifold is made of a patch of open sets $\{U_\alpha\}_\alpha$ that are ``local'' open sets. An algebra modelling the quantum base space is defined similarly as a patch of algebra $\{\mathcal{A}_\alpha\}_\alpha$, which elements ($\mathcal{A}_\alpha$) are called ``local''. In other words, the subscript $\alpha$ indicates that we are working in a "localised" region of (quantum) spacetime.
}
algebra $\mathcal{A}_\alpha$ is the noncommutative counterpart of the set of smooth functions $\mathcal{C}^\infty(U_\alpha)$. We further require that the local algebra $\mathcal{A}_\alpha$ are linked to $\kappa$-Minkowski, noted $\kM$, by a coordinate chart, so that the star-product of $\kM$ can be exported to a star-product in $\mathcal{A}_\alpha$.

\paragraph{} 
A missing ingredient is at the core of this paper: the tangent space. As recalled above the local tangent space relies (for pseudo-Riemannian manifold) on the Minkowski space through \eqref{eq:int_loc_tan}. Therefore, we construct, thanks to the coordinate charts, the the noncommutative analogue of \eqref{eq:int_loc_tan} (to be characterised), involving the $\kappa$-Minkowski spacetime. In order to do so, one needs to turn to the set of local vector fields, corresponding, in the noncommutative formalism, to the derivations of $\mathcal{A}_\alpha$. As explained in the ensuing discussion, \eqref{eq:int_loc_tan} is generalised as 
\begin{align}
    \Der(\mathcal{A}_\alpha) = \mathcal{Z(A_\alpha)} \otimes \kM,
    \label{eq:int_loc_der}
\end{align}
where $\mathcal{Z(A_\alpha)}$ denotes the centre of the algebra $\mathcal{A}_\alpha$. This is shown to have strong connections with the quantum fiber bundles defined in \cite{Brzezinski_1993}. Local forms are built in a dual way.

\paragraph{}
Finally, we introduce in this paper the notion of a noncommutative partition of unity. The aim is to define properly some global quantities, such as a global star-product, global derivations and forms, or a global integral, based on the local versions of these quantities. It completes the setting and allow to export the properties of $\kM$ to the whole quantum spacetime in a coherent way. We also ensure that these global objects do not depend on the choice of the partition nor of the covering of algebra thanks to a specific property of the partition of unity, already present in the commutative case.

\paragraph{}
The paper is organised as follows. First, we develop the case of the quantum analogue of a trivial tangent bundle in section \ref{sec:triv}. This is the noncommutative counterpart of the case where $\{U_\alpha\}_\alpha = \{\mathcal{M}\}$, \textit{i.e.}\ one may consider a single patch formed with the full manifold, for which global coordinates can be defined. The bridge between commutative and quantum quantities is made for the star-product, vector fields and forms. Note that one could consider the section \ref{sec:triv} as an application of the quantum fiber bundle of \cite{Brzezinski_1993} with $\kappa$-Minkowski as fibers, to some extent. We also define the notions of connection and curvature on the trivial bundle, which can be exported as such in the locally trivial version.

However, in the commutative setting not all spacetimes have a trivial tangent bundle, so that one has to define the notion of local triviality. Therefore, we export the notion of triviality to the one of local triviality in section \ref{sec:loc_triv}.

In the section \ref{sec:ncpu}, we introduce a quantum (noncommutative) version of the partition of unity. The commutative notion of partition of unity is recalled in section \ref{subsec:part_unit}. Its noncommutative version is defined and discussed in \ref{subsec:noncopar-comod}, together with a notion of covering of algebra. The use of the partition is made to define global objects in sections \ref{subsec:loc_triv_diff_calc} and \ref{subsec:loc2glob_quant}. In section \ref{sec:conc}, we discuss the results and conclude.

\paragraph{Notations and conventions:}
In the following, $\mathcal{M}$ denotes a smooth $d+1$-dimensional pseudo-Riemannian manifold with corresponding algebra of smooth functions $\mathcal{C^\infty(M)}$. The set of vector fields denoted by $\Gamma(\mathcal{M})$, corresponds to the Lie algebra of derivations acting on (smooth) functions. An atlas $\{(U_\alpha,\varphi^\alpha)\}_\alpha$ is as usual defined as a covering of open set $\{U_\alpha\}_\alpha$ of $\mathcal{M}$ and coordinate maps $\varphi_\alpha : U_\alpha \to \mathbb{R}^{1,d}$. Otherwise stated, the indices labelling the atlases $\alpha,\beta$ take values in a possibly infinite set and {\it{do not}} obey the usual summation convention. The spacetime indices $\mu,\nu,\lambda,\tau \in \{0, \dots, d\}$ follow summation convention of repeated indices. $\Omega^\bullet(\mathcal{M})$ denotes the exterior algebra of forms on $\mathcal{M}$. $\mathcal{A}$, $\mathcal{A}_\alpha$ refer to ${}^*$-algebras with involution denoted hereafter by ${}^*$. The indices for the covering of algebras $\alpha,\beta$ introduced in the course of the analysis follow the same rule as above. $\mathcal{Z(A)}$ denotes the cenre of $\mathcal{A}$. $\Der(\mathcal{A})$ and $\Omega^\bullet(\mathcal{A})$ denotes respectively the Lie algebra of derivations and the exterior algebra of forms of $\mathcal{A}$. The set of linear functions $f:\mathcal{A}\to \mathcal{B}$, where $ \mathcal{A}$ and $\mathcal{B}$ are two algebras, is denoted by $\Lin(\mathcal{A}, \mathcal{B})$, with the composition law $\circ$. Finally, the support of a function $\chi$ is denoted by $\Supp(\chi)$. The $\mathbb{C}$-linear span of a set $S$ is denoted $\mathrm{Span}_\mathbb{C}(S)$. Other additional notations are introduced throughout the paper.

To be definite, we consider in the sequel the $\kappa$-Minkowski space, denoted by $\kM$. Through this paper, the word ``quantum'' stands for noncommutative. We will use both indifferently. When dealing with Hopf algebra structures, the usual Sweedler notations is used, \emph{i.e.}\ one writes generically $\Delta(a) =  a_{(1)} \otimes a_{(2)}$ with summation understood.

\section{A quantum analogue of a trivial tangent bundle}
\label{sec:triv}

In this section, we focus, as a warm up, on a framework which may be viewed as a quantum (noncommutative) analogue of a \emph{trivial} tangent bundle.

\subsection{The star-product quantisation with global coordinates}
\label{subsec:gc_q}
\paragraph{}
We consider here the pseudo-Riemannian manifold $\mathcal{M}$ to have globally defined coordinates $\varphi : \mathcal{M} \to \Mink$, where $\varphi$ is a smooth diffeomorphism. One can define the pull-back\footnote{
    Note that usual notations of the pull-back are rather $\varphi^*$, but the ${}^*$ is dropped throughout the paper to avoid confusion with involutions.
} 
of $\varphi$ as $\varphi(f) = f \circ \varphi$, for any $f \in \mathcal{C}^\infty(\Mink)$, that is
\begin{align} 
    \varphi : \mathcal{C}^\infty(\Mink) \to \mathcal{C}^\infty(\mathcal{M}).
    \label{eq:qfst_pb}
\end{align}
One can further show that the pull-back is a homomorphism of algebra since $\varphi(fg) = \varphi(f) \varphi(g)$.

\paragraph{}
The noncommutative generalisation of the previous discussion considers $(\mathcal{A}, \star)$, an associative ${}^*$-algebra, to be the quantum version of the smooth functions $\mathcal{C}^\infty(\mathcal{M})$ and similarly $(\kM, \star_\kappa)$ to be the quantum version of $\mathcal{C}^\infty(\Mink)$. Then, we define a coordinate transformation as an invertible homomorphism 
\begin{align}
    \varphi : \kM \to \mathcal{A}
    \label{eq:qfst_qpb}
\end{align}
where \eqref{eq:qfst_qpb} is motivated by \eqref{eq:qfst_pb}. The fact that $\varphi$ is an invertible homomorphism allows one to relate the star-products of $\mathcal{A}$ and $\kM$ through
\begin{align}
    f \star g
    = \varphi^{-1} \big(\varphi(f) \star_\kappa \varphi(g) \big),
    \label{eq:qfst_loc_sp_rel}
\end{align}
for any $f, g \in \mathcal{A}$. The equation \eqref{eq:qfst_loc_sp_rel} states that the quantisation we impose on $\mathcal{A}$ is fully determined by the quantisation chosen for the Minkowski spacetime (taken here to be the so-called $\kappa$-deformation).

\paragraph{}
An important point needs to be made here. The previous expression \eqref{eq:qfst_loc_sp_rel} explicitly depends on the choice of $\varphi$, that is on the choice of coordinates. Therefore, considering another coordinate choice $\tilde{\varphi}$ generates another star-product $\tilde{\star}$ through \eqref{eq:qfst_loc_sp_rel}. Then, it is matter of algebra to show that
\begin{align}
    f\ \tilde{\star}\ g
    = (\varphi^{-1} \circ \tilde{\varphi})^{-1} \Big( (\varphi^{-1} \circ \tilde{\varphi})(f) \star (\varphi^{-1} \circ \tilde{\varphi})(g) \Big)
    \label{eq:qfst_eq_sp}
\end{align}
for any $f,g \in \mathcal{A}$. The previous expression make explicit the change of coordinates $\varphi^{-1} \circ \tilde{\varphi} : \mathcal{A} \to \mathcal{A}$. A consistency condition would be to require that the physical model build on $\mathcal{A}$ is invariant under any coordinate change of the form above, thus making contact with the diffeomorphism invariance in the commutative setting. Note that $\star$ and $\tilde{\star}$ are belonging to the same equivalent class of deformation of $\mathcal{A}$ (see for example Prop 2.14 of \cite{Bursztyn_2002}). In other words, a condition of diffeomorphism invariance imposes that the physical quantities remain the same for any class of equivalent star-products.

\paragraph{}
Finally, note that the existence of the previous $\varphi$ maps can be justified by formal deformation theory. Indeed, if we consider the local algebra to correspond to some formal power series of $\kappa$, that is $\mathcal{A} = \mathcal{C}^\infty(\mathcal{M})[[\kappa]]$, then one can define its product as
\begin{align}
    f \star g
    = fg + \sum_{n=0}^\infty \frac{1}{\kappa^n} \mu_n(f,g),
\end{align}
where $f,g \in \mathcal{A}$. Correspondingly, we can write the product on $\kappa$-Minkowski \eqref{eq:kM_star-kappa} in terms of formal power series
\begin{align}
    F \star_\kappa G
    = FG + \sum_{n=0}^\infty \frac{1}{\kappa^n} \tilde{\mu}_n(F,G),
\end{align}
for any $F,G \in \kM = \mathcal{C}^\infty(\Mink)[[\kappa]]$. In the previous expressions, $\mu_n$ and $\tilde{\mu}_n$ are bilinear operators that should be associative and satisfy the Bianchi identities. Therefore, the requirement that $\varphi$ is a homomorphism amounts to
\begin{align}
    \mu_n(f, g)
    = \varphi^{-1} \circ \tilde{\mu}_n \big( \varphi(f), \varphi(g) \big).
    \label{eq:qfst_formal_pb}
\end{align}
In the case where the $\tilde{\mu}_n$ contains only derivatives of the functions (which is the case of the star-product \eqref{eq:kM_star-kappa}), then the commutative pull-back $\varphi$ of \eqref{eq:qfst_pb} satisfies \eqref{eq:qfst_formal_pb}.

\subsection{From global coordinates to triviality}
\label{subsec:gctt}
\paragraph{}
In a commutative situation (see e.g.\ \cite{Gockeler_1987, Rudolph_2012, Rudolph_2017}), the tangent bundle $T\mathcal{M}$, of a pseudo-Riemannian manifold $\mathcal{M}$, is trivial if it is of the form
\begin{equation}
    T\mathcal{M} \simeq \mathcal{M} \times \Mink.
    \label{eq:triv_tangent_space}
\end{equation}
As a vector bundle, $T\mathcal{M}$ can be conveniently characterised with the help of the set of vector fields $\Gamma(\mathcal{M})$. Now, the triviality \eqref{eq:triv_tangent_space} of $T\mathcal{M}$ is equivalent to the existence of a globally defined frame, says $e \in L\mathcal{M}$ ($L\mathcal{M}$ denotes the frame bundle) such that for any $x\in\mathcal{M}$, the family of (global) sections $\{e_\mu\}_\mu$, $e_\mu:\mathcal{M}\to T\mathcal{M}$, is a basis of $\Mink$. Accordingly, any vector field $X\in\Gamma(\mathcal{M})$, $X : \mathcal{M} \to T\mathcal{M}$, can be expressed as $X = X^\mu e_\mu$ where $X^\mu \in \mathcal{C^\infty(M)}$ which yields to the following identification
\begin{align}
    \Gamma(\mathcal{M}) 
    \simeq \mathcal{C^\infty(M)} \otimes \Mink,
    \label{eq:triv_vector_fileds_sim}
\end{align}
with of course any $e_\mu(x)$ acting as a derivation of $\mathcal{C^\infty(M)}$, so that one can actually write $\Gamma(\mathcal{M})=\Der(\mathcal{C^\infty(M)})$. By using standard duality, \eqref{eq:triv_tangent_space} translates into
\begin{equation}
    \Omega^1(\mathcal{M}) \simeq \mathcal{C^\infty(M)} \otimes \Mink,
    \label{eq:commut-1form}
\end{equation}
where $\Omega^1(\mathcal{M})$ denotes as usual the space of 
1-forms on $\mathcal{M}$.

\paragraph{}
In order to extend the above framework to a noncommutative setting, it is convenient to start from the notion of vector field acting on $\mathcal{C^\infty(M)}$. As noncommutative analogue of this object, we choose the notion of derivation of $\mathcal{A}$, noted $\Der(\mathcal{A})$. This is known to be a natural choice.

Having in mind future application to gauge theories on quantum spaces, it is tempting to replace the factor $\Mink$ in \eqref{eq:triv_tangent_space} and \eqref{eq:triv_vector_fileds_sim} by the $\kappa$-Minkowski space $\kM$, which now plays, roughly speaking, the role of ``noncommutative  (quantum)'' tangent space. In this respect, the algebra $\mathcal{A}$ introduced above plays the role of a ``noncommutative  (quantum)'' base space, provided a suitable action of a set of elements of $\kM$ as derivation of $\mathcal{A}$ can be defined.

Various approaches toward a suitable notion of quantum fiber bundle have been elaborated and used, see \cite{Brzezinski_1993}, \cite{Schneider}, \cite{Pflaum_1994} (see also \cite{Aschieri_2021b}, \cite{Aschieri_2021c}). We adapt these in this paper to set-up a natural framework implementing the above scheme.

\paragraph{}
Recall that $\kM$ is defined as the dual Hopf algebra of the deformed translations, \textit{i.e.}\ $\kM = \mathcal{T}_\kappa^\prime$, where $\mathcal{T}_\kappa $ is a Hopf subalgebra of the $\kappa$-Poincar\'e algebra $\kP$ (see the Appendix \ref{ap:kM} for conventions). This duality allows to compute the full Hopf algebra structure of $\kM$ \cite{MR1994}. Indeed, $\kM$ is generated by $\{p_\mu\}_{\mu \in \{0, \dots, d\}}$ with
\begin{subequations}
\begin{align}
    [p_0, p_j] &= \frac{i}{\kappa} p_j, &
    [p_j, p_k] &= 0, 
    \label{eq:kM_kM_Hopf_alg_alg}\\
    \Delta(p_\mu) &= p_\mu \otimes 1 + 1 \otimes p_\mu, &
    S(p_\mu) &= - p_\mu,
    \label{eq:kM_kM_Hopf_alg_coalg}
\end{align}
    \label{eq:kM_kM_Hopf_alg}
\end{subequations}
where $[f, g] = f \star_\kappa g - g \star_\kappa f$ and $\star_\kappa$ is the associative star-product of $\kM$. The Lie-algebra structure carried by $\kM$ is apparent from \eqref{eq:kM_kM_Hopf_alg_alg} while its Hopf-algebra structure is obvious from \eqref{eq:kM_kM_Hopf_alg_coalg}. Note that we denote the generators of $\kM$ by $p_\mu$ instead of the traditional notation $x_\mu$ prevailing in the physics literature in view of the role they play.

For convenience of notations, let $ \mathfrak{D}_\kappa$ denotes the following $\mathbb{C}$-linear space of elements of $\kM$, 
\begin{equation}
    \mathfrak{D}_\kappa := \mathrm{Span}_\mathbb{C} \big\{ p_\mu, \mu\in \{0, \dots, d\} \big\}.
    \label{eq:triv_Dk_def}
\end{equation}
Similarly, we note $\mathfrak{D}_\kappa'$ its dual Hopf algebra consisting of elements of $\mathcal{T}_\kappa$. It is characterised by 
\begin{equation}
    \mathfrak{D}_\kappa^\prime = \mathrm{Span}_\mathbb{C} \big\{\mathfrak{X}^\mu, \mu\in\{0, \dots, d\} \big\},
\end{equation}
with $\mathfrak{X}^0 = \kappa(1-\mathcal{E})$, $\mathfrak{X}^j = P^j,\ \  j\in\{1, \dots, d\}$, the generators of $\mathcal{T}_\kappa$ (see Appendix \ref{ap:kM}).

\paragraph{}
Consider now that $\mathcal{M}$ has global coordinates as in section \ref{subsec:gc_q}. This allows to define a coaction $\coactr: \mathcal{A} \to \mathcal{A} \otimes \mathfrak{D}_\kappa'$. The coaction we consider here writes
\begin{align}
    \coactr f
    &= \varphi^{-1} \big( \mathfrak{X}_\mu \actl_\kappa \varphi(f) \big) \otimes \mathfrak{X}^\mu
    \label{eq:gc_kTcoact}
\end{align}
where $\actl_\kappa$ is the action of $\kP$ on $\kM$ \eqref{eq:kM_kP_action}. Thus, the algebra $\mathcal{A}$ becomes a $\mathcal{T}_\kappa$-comodule algebra%
\footnote{%
    Given a Hopf algebra $H$ with coproduct $\Delta$ and counit $\epsilon$, recall that a right $H$-comodule algebra $\mathcal{A}$ is an algebra with coaction map $\coactr : \mathcal{A} \to \mathcal{A}\otimes H$ satisfying $(\coactr \otimes \text{id}_H) \circ \coactr = (\text{id}_\mathcal{A} \otimes \Delta) \coactr$, $(\text{id}_\mathcal{A} \otimes \epsilon) \coactr = \text{id}_\mathcal{A} $ with $\coactr$ being an algebra homomorphism.
}.
For details on comodule algebras and related topics, see e.g.\ \cite{Montgomery_1993} (Chapter 4) or \cite{Delvaux_2009}. This implies that $\mathcal{A}$ is a $\kM$-module algebra%
\footnote{%
    Given a Hopf algebra $H$ with coproduct $\Delta$ and counit $\epsilon$, recall that a left $H$-module algebra $\mathcal{A}$ is an algebra with action map $\actl : H \otimes \mathcal{A} \to \mathcal{A}$ satisfying $\actl \circ (\text{id}_H \otimes m) = m \circ (\actl \otimes \actl) \circ (\text{id}_H \otimes \tau \otimes \text{id}_\mathcal{A}) \circ (\Delta \otimes \text{id}_\mathcal{A} \otimes \text{id}_\mathcal{A})$ and $\actl \circ (\text{id}_H \otimes 1_\mathcal{A}) = 1_\mathcal{A} \circ \epsilon$ where $m:\mathcal{A} \otimes \mathcal{A} \to \mathcal{A}$ is the product on $\mathcal{A}$, $1_\mathcal{A}: \mathbb{C} \to \mathcal{A}$ is the unit of $\mathcal{A}$ and $\tau : H \otimes \mathcal{A} \to \mathcal{A} \otimes H$ is the flip map. This is equivalent to $m$ and $1_\mathcal{A}$ being left $H$-module homomorphisms.
}
with an action given for $p_\mu \in \mathfrak{D}_\kappa$ a generator and any $f \in \mathcal{A}$, by%
\footnote{%
    Consider $H$ to be a Hopf algebra, $H'$ its dual Hopf algebra with dual pairing $\langle \cdot, \cdot \rangle : H \times H' \to \mathbb{C}$, and $\mathcal{A}$ a right $H$-comodule algebra, with coaction $\coactr : \mathcal{A} \to \mathcal{A} \otimes H$. Then $\mathcal{A}$ is a left $H'$-module algebra with the action defined as
    \begin{align}
        p \actl f
        = \langle f_{\langle 2\rangle}, p \rangle f_{\langle 1 \rangle}, &&
        p \in H',\ f \in \mathcal{A},
        \label{eq:modalg_loc_act}
    \end{align}
    where we used the Sweedler notations $\coactr f = f_{\langle 1 \rangle} \otimes f_{\langle 2 \rangle}$ (so that $f_{\langle 2 \rangle} \in H$).
}
\begin{align}
    p_\mu \actl f
    &= \varphi^{-1} \big( \mathfrak{X}_\mu \actl_\kappa \varphi (f) \big)
    = \varphi^{-1} \left( \partial^\mu \varphi(f) \right)
    \label{eq:gc_kMact}
\end{align}
where $\partial^\mu = \frac{\partial}{\partial p_\mu}$ denotes the derivation with respect to the $\mu$-th coordinate in $\kappa$-Minkowski.

\subsection{\tops{$\kappa$}{kappa}-Minkowski space as a quantum tangent space}
\label{subsec:triv_tangent_sp}
\paragraph{}
Let $\Der(\mathcal{A})$ denotes the set of derivations of $\mathcal{A}$. For $\kM$ to be promoted as a constituent of a notion of quantum tangent space, it seems natural to require that (some of) its elements act as derivations of $\mathcal{A}$ in the spirit of the derivation-based differential calculus framework (for a review see \cite{Dubois_2002}).

\paragraph{}
First, one has the following useful properties satisfied by the linear structure carried by $\kM$.
\begin{proposition}\label{prop21}
    One has the following:
    \begin{enumerate}[label=\roman*)]
        \item \label{it:triv_Dk_der}
        $\mathfrak{D}_\kappa \subset \Der(\mathcal{A})$ and this inclusion is injective,
        \item \label{it:triv_Dk_Lie}
        $\mathfrak{D}_\kappa$ is a Lie algebra when endowed with the Lie bracket \eqref{eq:kM_kM_Hopf_alg_alg},
        \item \label{it:triv_Dk_mod}
        $\mathfrak{D}_\kappa$ is a $\mathcal{Z}(\mathcal{A})$-module for the product defined by $(zp_\mu)\actl f=z(p_\mu\actl f)$, for any $z\in\mathcal{Z(A)}$, $p_\mu\in\mathfrak{D}_\kappa$, $f\in\mathcal{A}$.
    \end{enumerate}
\end{proposition}
\begin{proof}
    From \eqref{eq:gc_kMact}, $\mathcal{A}$ is a left $\kM$-module algebra. Therefore, using Sweedler notations, $p_\mu \actl (fg) = \sum (p_{\mu(1)} \actl f) (p_{\mu(2)} \actl g)$, which combined with \eqref{eq:kM_kM_Hopf_alg_coalg} yields
    \begin{align}
	   p_\mu \actl (fg) 
	   = (p_{\mu(1)} \actl f) (p_{\mu(2)} \actl g)
	   = (p_\mu \actl f) g + f (p_\mu \actl g),
	   \label{eq:triv_module_alg_def}
    \end{align}
    for any $f,g \in\mathcal{A}$, where the first equality expresses the fact that the product in $\mathcal{A}$ is a left $\kM$-module homomorphism and the second equality is implied by \eqref{eq:kM_kM_Hopf_alg_coalg}. From the rightmost equality in \eqref{eq:triv_module_alg_def}, one concludes that any $p_\mu\in\mathfrak{D}_\kappa$ obeys a Leibnitz rule. Hence, $ \mathfrak{D}_\kappa \subset \Der(\mathcal{A})$. Moreover, let us define the coordinate function
    \begin{align}
        x^\mu = \varphi^{-1}(p_\mu).
        \label{eq:gc_coord_def}
    \end{align}
    Consider $\alpha^\mu p_\mu, \beta^\nu p_\nu \in \mathfrak{D}_\kappa$, with $\alpha^\mu, \beta^\nu \in \mathbb{C}$, that satisfy $(\alpha^\mu p_\mu) \actl f = (\beta^\nu p_\nu) \actl f$, for any $f \in \mathcal{A}$. Then, $(\alpha^\mu - \beta^\mu) (p_\mu \actl f) = 0$ and considering $f = x^\nu$ of \eqref{eq:gc_coord_def} together with the definition \eqref{eq:gc_kMact}, one obtains that $\alpha^\nu = \beta^\nu$. This directly implies that $\mathfrak{D}_\kappa \to \Der(\mathcal{A})$ with $p_\mu \mapsto p_\mu \actl \cdot$ is injective. Therefore, \ref{it:triv_Dk_der} is proven.
    
    Now, for any $z \in \mathcal{Z}(\mathcal{A})$, $p_\mu \in \mathfrak{D}_\kappa$, simply write $(z p_\mu) \actl (f g) = z( (p_\mu \actl f) g + f (p_\mu \actl g) ) = (z p_\mu) \actl f) g + f z (p_\mu \actl g) = ((z p_\mu) \actl f) g + f ((z p_\mu) \actl g)$ where we used $z f = f z$ in the second equality. Hence $\mathfrak{D}_\kappa$ is a $\mathcal{Z}(\mathcal{A})$-module as claimed in \ref{it:triv_Dk_mod}. 
    
    Finally, the Lie algebra structure of $\mathfrak{D}_\kappa$ claimed in \ref{it:triv_Dk_Lie} stems simply from \eqref{eq:kM_kM_Hopf_alg_alg}.
\end{proof}

At this stage, two comments are in order.
\begin{enumerate}
    \item Choosing $\mathfrak{D}_\kappa$, defined by \eqref{eq:triv_Dk_def}, in \eqref{eq:derA_def} instead of the full $\kM$ \eqref{eq:kM_kM_Hopf_alg} is motivated by the fact that $\mathfrak{D}_\kappa$ has a finite number of generators as a Lie algebra, contrary to $\kM$ which has infinitely many. As far as physical applications are concerned, a differential calculus based on the full $\kM$ would then give rise to gauge potentials and curvature with infinite number of components. 
    \item Even if Proposition \ref{prop21} uses that $\mathcal{A}$ is a $\mathfrak{D}_\kappa$-module algebra, we wished to underline that it can be inherited form a $\mathfrak{D}_\kappa'$-comodule algebra structure \eqref{eq:gc_kTcoact}. This proves useful to make some connection with the framework given in \cite{Brzezinski_1993}. 
\end{enumerate}

We are now in position to define a reasonable quantum analogue of the space of vector fields $\Gamma(\mathcal{M})$. In fact, the above discussion and the structure of eqn.\ \eqref{eq:triv_vector_fileds_sim} point toward
$\mathcal{Z(A)} \otimes\mathfrak{D}_\kappa\subset\mathcal{Z(A)} \otimes \kM$ as a natural quantum version of $\Gamma(\mathcal{M})$. By a slight abuse of language, it is called the quantum space of vector fields in this paper, although its elements are not {\it{stricto sensu}} vector fields. Note that the occurrence of the factor $\mathcal{Z(A)}$, instead of $\mathcal{A}$, insure as usual that the corresponding subset of derivations is a $\mathcal{Z(A)}$-module. Thus, we set the following:

\begin{definition}\label{def:derA_def}
Define $\Der_R(\mathcal{A}) := \mathcal{Z(A)} \otimes \mathfrak{D}_\kappa$, the quantum space of vector fields, as
\begin{align}
	\Der_R(\mathcal{A}) 
	:= \mathrm{Span}_\mathbb{C} \big\{z\otimes p_\mu\in\mathcal{Z(A)} \otimes \mathfrak{D}_\kappa\ /\ (z \otimes p_\mu) (f)
	= z(p_\mu \actl f),\ \ \forall z\in\mathcal{Z(A)}, \ \forall f\in\mathcal{A} \big\}
\label{eq:derA_def}
\end{align}
\end{definition}

\begin{proposition}\label{prop:derA_prop}
    The following properties hold true.
    \begin{enumerate}[label = \roman*)]
        \item \label{it:triv_coDk_der}
        $\Der_R(\mathcal{A}) \subset \Der(\mathcal{A})$ and this inclusion is injective,
        \item \label{it:triv_coDk_mod}
        $\Der_R(\mathcal{A})$ is a $\mathcal{Z(A)}$-module,
        \item \label{it:triv_coDk_Lie}
        $\Der_R(\mathcal{A})$ has a Lie algebra structure when equipped with the following Lie bracket $[X,Y](f) = X(Y(f)) - Y(X(f))$ for any $X, Y \in \Der_R(\mathcal{A})$, $f \in \mathcal{A}$.
    \end{enumerate}
\end{proposition}

\begin{proof}
    First, observe that the linear structure of $\mathfrak{D}_\kappa$ extends to $\Der_R(\mathcal{A})$. Indeed, write any element of $\mathfrak{D}_\kappa$ as $p=\lambda^\mu p_\mu$ with $\lambda^\mu \in \mathbb{C}$. Upon using the canonical identification $\mathcal{Z(A)} \otimes \mathbb{C} \simeq \mathcal{Z(A)}$, one can write for any $X\in\Der_R(\mathcal{A})$,
    \begin{equation}
        X
        = Z \otimes p
        = Z \otimes \lambda^\mu p_\mu
        = Z^\mu p_\mu
    \end{equation}
    with $Z^\mu = Z \otimes \lambda_\mu \in \mathcal{Z(A)}$.
    
    By using Proposition \ref{prop21}, one computes that for any $z\in\mathcal{A}$, $f, g \in\mathcal{A}$,
    \begin{align}
	   (z \otimes p_\mu) (f g)
	   &= z (p_\mu \actl (f g))
	   = z(p_\mu \actl f)vg + z f (p_\mu \actl g) \label{deca1}\\
	   \big( (z \otimes p_\mu) (f) \big) g + f \big( (z \otimes p_\mu) (g) \big)
	   &= z(p_\mu \actl f) g + f z (p_\mu \actl g)\label{deca2},
    \end{align}
    where \eqref{eq:triv_module_alg_def} and \eqref{eq:derA_def} have been used, so that the RHS of \eqref{deca1} and \eqref{deca2} are equal because $z \in \mathcal{Z(A)}$ by assumption. Therefore, one concludes that any element of $\Der_R(\mathcal{A})$ actually corresponds to a derivation of $\mathcal{A}$. Hence $\Der_R(\mathcal{A}) \subset\Der(\mathcal{A}) $ as claimed in \ref{it:triv_coDk_der}. The injectivity of the map $\Der_R(\mathcal{A}) \to \Der(\mathcal{A})$, $Z^\mu \otimes p_\mu \mapsto (Z^\mu p_\mu) \actl \cdot$ follows the same steps as in Proposition \ref{prop21}.
    
    Next, note that one has the following natural action 
    \begin{equation}
        z \actl (Z^\mu \otimes p_\mu)
	    = (z Z^\mu) \otimes p_\mu,\quad (Z^\mu \otimes p_\mu) \actr z
	    = (Z^\mu z) \otimes p_\mu.\label{zeactiononforms}
    \end{equation}
    for any $z, X^\mu \in \mathcal{Z(A)}$. Thus, the $\mathcal{Z(A)}$-module structure of $\Der_R(\mathcal{A})$ is a straightforward consequence of Proposition \ref{prop21}.
    
    Finally, a standard computation give rise, for any $X = Z_1^\mu p_\mu, Y = Z_2^\nu p_\nu \in\Der_R(\mathcal{A}) $, $Z_1^\mu, Z_2^\nu \in \mathcal{Z(A)}$, to
    \begin{equation}
        [Z_1^\mu \otimes p_\mu, Z_2^\nu \otimes p_\nu]
        = Z_1^\mu p_\mu(Z_2^\nu) \otimes p_\nu - Z_2^\nu p_\nu(Z_1^\mu) \otimes p_\mu + Z_1^\mu Z_2^\nu \otimes [p_\mu, p_\nu]
    \end{equation}
    where the RHS is thus a linear combination of elements of $\Der_R(\mathcal{A}) $ and the last claim follows. 
\end{proof}

Some remarks are in order.
\begin{enumerate}
    \item Note that $\Der_R(\mathcal{A})$ with its action defined in \eqref{eq:derA_def} is a $\mathcal{Z(A)}$-module, but not a $\mathcal{A}$-module. Indeed, if one defines the action $(Z^\mu \otimes p_\mu) (a):= Z^\mu (p_\mu \actl a)$ for any $Z^\mu \in \mathcal{A}$, it implies that 
    \begin{equation}
        (Z_1^\mu \otimes p_\mu) \actl \Big((Z_2^\nu \otimes p_\nu) \actl f \Big) - (Z_2^\nu \otimes p_\nu) \actl \Big((Z_1^\mu \otimes p_\mu) \actl f \Big) \ne [Z_1^\mu \otimes p_\mu, Z_2^\nu \otimes p_\nu] \actl f
    \end{equation}
    for any $f \in \mathcal{A}$ as a consequence of the Leibniz rule. 
    \item The fact that $\Der_R(\mathcal{A}) \subset\Der(\mathcal{A}) $ implies to consider only a limited set of derivations. This gives rise to a so called restricted derivation-based differential calculus, used in particular to construct noncommutative gauge theories on various quantum spaces, see e.g. \cite{wal-wulk1}-\cite{gauge-kappa1}. For a related review, see \cite{physrep}. 
    \item One observes that the linear space $\mathcal{A}\otimes\kM$ can be endowed with the structure of right $\mathcal{T}_\kappa$-comodule algebra when equipped with the smash product and extended coaction. Indeed, the following property holds:
    \begin{proposition}\label{smash-prop1}
        Let us define the smash product $E = \mathcal{A} \# \kM := \mathcal{A} \otimes \kM$ as a vector space, with the product
        \begin{align}
            (f \# p) (g \# q)
            = f (p_{(1)}\actl g) \# (p_{(2)} q)
        \end{align}
        for any $p \in \kM$ and $f, g \in \mathcal{A}$. $E$ equipped with the coaction $\coactr : E \to E \otimes \kM$, $\coactr (f \otimes p) = f \otimes p_{(1)} \otimes p_{(2)}$ is a right $\kM$-comodule algebra.
    \end{proposition}
    \begin{proof}
        This is a direct consequence of Lemma 7.1.2 of ref. \cite{Montgomery_1993} and the fact that $\mathcal{A}$ is a $\kM$-module algebra.
    \end{proof}
    From the last remark, it follows that the $\mathcal{T}_\kappa$-comodule structure of the quantum base space $\mathcal{A}$ transfers naturally to $E$ as a $\kM$-comodule algebra. In view of \eqref{eq:triv_tangent_space} and refs. \cite{Brzezinski_1993}-\cite{Pflaum_1994}, oine could also consider $E = \mathcal{A}\#\kM$ (see Proposition \ref{smash-prop1}) as the quantum version of the trivial tangent bundle related to the quantum base space $\mathcal{A}$.
\end{enumerate}

\subsection{Restricted differential calculus and linear connection}
\label{subsec:triv_diff_calc}
\paragraph{}
We now turn to the characterisation of a restricted noncommutative differential calculus. Several notions of noncommutative differential calculus have already been developed in the literature. We focus here on a convenient one introduced in \cite{Dubois_Violette_1996} that exploits the notion of homomorphisms of module.

Explicitly, given an associative ${}^*$-algebra $\mathcal{A}$, the space of 1-forms $\Omega^1(\mathcal{A})$ can be defined as $\Omega^1(\mathcal{A}) = \Hom_{\mathcal{Z(A)}} \big(\Der(\mathcal{A}), \mathcal{A} \big)$, that is the set of homomorphisms of $\mathcal{Z(A)}$-modules from $\Der(\mathcal{A})$ to $\mathcal{A}$. This generalisation of one-forms is included in the formalism of derivation-based differential calculus, see e.g \cite{Dubois_Violette_1996}. Recall that the latter defines the set of $n$-forms as $\mathcal{Z(A)}$-multilinear antisymmetric maps from $\Der(\mathcal{A})^n$ to $\mathcal{A}$, together with the wedge product \eqref{eq:triv_diff_calc_wedge} and a differential \eqref{eq:triv_diff_calc_diff} defined below. In this section, we built the so-called restricted derivation based differential calculus by merely replacing $\Der(\mathcal{A})$ by $\Der_R(\mathcal{A})$.

\paragraph{}
In the above setting, as $\mathcal{Z(A)} \otimes \mathfrak{D}_\kappa$ plays the role of a quantum version of $\Gamma(\mathcal{M})$, one can define, for $n \in \mathbb{N}$, $\Omega_R^n(\mathcal{A})$ as the set of $\mathcal{Z(A)}$-multilinear antisymmetric maps from $(\mathcal{Z(A)} \otimes \mathfrak{D}_\kappa)^n$ to $\mathcal{A}$, with $\Omega^0(\mathcal{A}) = \mathcal{A}$.

\paragraph{}
Accordingly, the wedge product of forms $\wedge : \Omega^n(\mathcal{A}) \times \Omega^m(\mathcal{A}) \to \Omega^{n+m}(\mathcal{A})$ is then defined through the usual formula 
\begin{equation}
	\begin{aligned}
		(\rho \wedge \eta) & (X_1, \dots, X_{n+m}) = \\
		& \frac{1}{n!m!} \sum_{\sigma \in \mathfrak{S}_{n+m}} (-1)^{\sign(\sigma)} \rho(X_{\sigma(1)}, \dots, X_{\sigma(n)}) \eta(X_{\sigma(n+1)}, \dots, X_{\sigma(n+m)}),
	\end{aligned}
	\label{eq:triv_diff_calc_wedge}
\end{equation}
where $\mathfrak{S}_{n}$ denotes the set of permutations of $n$ elements. Eqn.\ \eqref{eq:triv_diff_calc_wedge} holds for any $X_1, \dots, X_{n+m} \in \Der(\mathcal{A})$ and any $\rho\in\Omega^n(\mathcal{A}) $, $\eta\in\Omega^m(\mathcal{A})$.

\paragraph{}
The differential $\dd : \Omega^{n}(\mathcal{A}) \to \Omega^{n+1}(\mathcal{A})$ is defined through the Koszul formula 
\begin{equation}
	\begin{aligned}
		\dd \rho(X_1, \dots, X_{n+1}) 
		&= 
		\sum_{j=1}^{n+1} (-1)^{j+1} X_j \rho(X_1, \dots, \omitel{j}, \dots, X_{n+1}) \\
		&+ 
		\sum_{1\leqslant j < k\leqslant n+1} (-1)^{j+k} \rho([X_j,X_k], X_1, \dots, \omitel{j}, \dots, \omitel{k}, \dots, X_{n+1}).
	\end{aligned}
	\label{eq:triv_diff_calc_diff}
\end{equation}
where $\omitel{j}$ corresponds to the omission of $X_j$.

\paragraph{}
The above objects can also be defined in the restricted case, that is on $\Omega_R(\mathcal{A})$, with derivations in $\Der_R(\mathcal{A})$. Note that the Koszul formula \eqref{eq:triv_diff_calc_diff} implies that for any $f \in \mathcal{A}$, $\dd f(p_\mu) = p_\mu \actl f$.

The restricticted one forms are thus defined as $\Omega_R^1(\mathcal{A}) = \Hom_{\mathcal{Z(A)}}(\mathfrak{D}_\kappa \otimes \mathcal{Z(A)}, \mathcal{A})$.
In view of $\Der_R(\mathcal{A}) = \mathcal{Z(A)} \otimes \mathfrak{D}_\kappa$, one can check that $\Omega^1_R(\mathcal{A}) \simeq \mathcal{A} \otimes \mathfrak{D}_\kappa^\prime$, which holds true as isomorphisms of $\mathcal{A}$-modules. In this property, the action of $\mathcal{A}$ on $\Hom_{\mathcal{Z(A)}}(\mathfrak{D}_\kappa \otimes \mathcal{Z(A)}, \mathcal{A})$ is defined by the product of $\mathcal{A}$, and the action on   $\mathcal{A} \otimes \mathfrak{D}_\kappa^\prime$ is given by 
\begin{align}
	f \actl (\mathfrak{a}_\mu \otimes \mathfrak{X}^\mu) 
	&=  f \mathfrak{a}_\mu \otimes \mathfrak{X}^\mu &
	(\mathfrak{a}_\mu \otimes \mathfrak{X}^\mu) \actr f
	&= \mathfrak{a}_\mu f \otimes \mathfrak{X}^\mu.
    \label{eq:triv_form_act}
\end{align}
for any $\mathfrak{a}_\mu, f \in\mathcal{A}$. This generates a differential calculus in the sense of \cite{Woronoviec_1989}.

\paragraph{}
Within this framework, a noncommutative version of a linear connection \cite{Dubois_2002} can be defined as a map $\nabla: \Der_R(\mathcal{A})\otimes \Der_R(\mathcal{A})\to\Der_R(\mathcal{A})$ such that for any $X,Y \in \Der_R(\mathcal{A})$ and $z \in \mathcal{Z(A)}$,
\begin{subequations}
\begin{align}
    \nabla_X(z \actl Y) 
	&= z \actl \nabla_X(Y) + X(z) \actl Y, &
	\nabla_X(Y \actr z) 
	&= \nabla_X(Y) \actr z + Y \actr X(z),
    \label{eq:conn2_act_leibniz} \\
	\nabla_{z \actl X}(Y)
	&= z \actl \nabla_X(Y), &
    \nabla_{X \actr z}(Y)
    &= \nabla_X(Y) \actr z,
	\label{eq:conn2_act_lin}
\end{align}
\end{subequations}
where the action is defined by \eqref{zeactiononforms}, which is supplemented by the following hermiticity condition:
\begin{equation}
    (\nabla_X(Y))^* = \nabla_{X^*}(Y^*).
    \label{hermit1}
\end{equation}
From this, using $X=X^\mu p_\mu$ and $Y=Y^\nu p_\nu$, one easily obtains
\begin{align}
    \nabla_{X^\mu p_\mu}(Y^\nu p_\nu)
	= X^\mu Y^\nu\Gamma_{\mu\nu}^\lambda p_\lambda
    + \big(X^\mu (p_\mu \actl Y^\nu) \big) p_\nu,
\label{eq:conn2_comp}
\end{align}
where we have set
\begin{equation}
    \nabla_{p_\mu}(p_\nu)
    := \Gamma_{\mu\nu}^\lambda \, p_\lambda
    \in \mathcal{Z(A)} \otimes \mathfrak{D}_\kappa.
\end{equation}
The theory of module algebra states, in the ${}^*$-algebra case that $(p \actl f)^* = S(p)^\dagger \actl f^*$, where $*$ is the involution in $\mathcal{A}$ and $\dagger$ the one in $\kM$. From the expression \eqref{eq:kM_kM_Hopf_alg_coalg} of the antipode of $\kM$, it implies that $X^*(f) := X(f^*)^*= - (X^\mu)^* (p_\mu \actl f)$, so that the condition \eqref{hermit1} becomes
\begin{equation}
    (\Gamma_{\mu\nu}^\lambda)^* = - \Gamma_{\mu\nu}^\lambda.
    \label{eq:triv_conn_hermit}
\end{equation}
The curvature is defined as a map $ R: \Der_R(\mathcal{A})\times \Der_R(\mathcal{A})\to \Der_R(\mathcal{A}) $ with
\begin{align}
	R(X,Y)(Z) = [\nabla_X, \nabla_Y](Z) - \nabla_{[X,Y]}(Z).
\label{eq:curv2_def}
\end{align}
From \eqref{eq:curv2_def}, one obtains
\begin{align}
	\tensor{R}{_{\mu\nu\lambda}^\tau}
	&= 
	(p_\mu \actl \Gamma_{\nu\lambda}^\tau) 
	- (p_\nu \actl \Gamma_{\mu\lambda}^\tau)
	+ \Gamma_{\nu\lambda}^\sigma \Gamma_{\mu\sigma}^\tau 
	- \Gamma_{\mu\lambda}^\sigma \Gamma_{\nu\sigma}^\tau
	- \tensor{C}{_{\mu\nu}^\sigma} \Gamma_{\sigma\lambda}^\tau,
\label{eq:curv2_comp}
\end{align}
where $	\tensor{C}{_{\mu\nu}^\lambda} = (\delta_\mu^0 \delta_\nu^\lambda - \delta_\nu^0 \delta_\mu^\lambda) \frac{i}{\kappa}$ is the structure constant of $\kM$, which can be read from \eqref{eq:kM_kM_Hopf_alg_alg}.

\section{Quantum analogue of locally trivial tangent bundle}
\label{sec:loc_triv}
\paragraph{}
In section \ref{sec:triv}, we proposed a framework attempting to set-up a noncommutative extension of a trivial tangent bundle. The purpose of this section is to modify this framework in order to set-up a noncommutative analogue of a {\it{locally trivial}} tangent bundle.

The main difficulty to overcome is to deal with the notion of locality which, {\it{stricto sensu}} simply disappears in a noncommutative framework.

\subsection{Covering of algebra with ideals}
\label{subsec:loct_ideal}
\paragraph{}
The notion of locality we want to generalise is strongly linked with the defining notion of open covers for manifolds. Therefore, we adopt here the generalisation to covering of algebra through ideals developed in \cite{Budzynski_1994, Calow_1999, Calow_2000a, Calow_2000b, Calow_2000c}.

\paragraph{}
Let $\mathcal{A}$ be an associative ${}^*$-algebra which corresponds to the ``base-space algebra''. Let $\{J_\alpha\}_\alpha$ be a finite family of $*$-ideal of $\mathcal{A}$.
\begin{definition}
    \label{def:cov_alg}
    $\{J_\alpha\}_\alpha$ is called a covering of $\mathcal{A}$ if
    \begin{align}
        \bigcap_\alpha J_\alpha = \{0\}.
        \label{eq:cov_alg_def}
    \end{align}
\end{definition}
From this, one can define the quotient algebras of $\mathcal{A}$ as $\mathcal{A}_\alpha = \mathcal{A} / J_\alpha$ and $\mathcal{A}_{\alpha\beta} = \mathcal{A} / (J_\alpha + J_\beta)$, where $J_\alpha + J_\beta$ corresponds to the smallest closed $*$-ideal containing $J_\alpha$ and $J_\beta$. By definition any $\mathcal{A}_\alpha$ is an algebra and we can turn it into a ${}^*$-algebra by defining 
\begin{align}
    \pi_\alpha(f)^* = \pi_\alpha(f^*),
    \label{eq:ideal_invol}
\end{align}
for $f \in \mathcal{A}$, where $\pi_\alpha : \mathcal{A} \to \mathcal{A}_\alpha$ be the surjective canonical homomorphism, \textit{i.e.}\ $\pi_\alpha(f) \in \mathcal{A}_\alpha$ is the representative of the equivalent class of $f \in \mathcal{A}$. The set $\{\mathcal{A}_\alpha\}_\alpha$ is called the covering of algebra of $\mathcal{A}$.

\paragraph{}
Observe that this setting consistently encompasses the commutative situation. Indeed, consider $\mathcal{A} = \mathcal{C^\infty(M)}$, where $\mathcal{M}$ is a differential manifold. Then, let $\{U_\alpha\}_\alpha$ be a covering of $\mathcal{M}$ and the map $|_\alpha : \mathcal{C^\infty(M)} \to \mathcal{C}^\infty(U_\alpha)$ be the restriction to $U_\alpha$. The kernel $\Ker(|_\alpha)$ is an ideal of $\mathcal{C^\infty(M)}$ and one has $C^\infty(U_\alpha) = \mathcal{C^\infty(M)} / \Ker(|_\alpha)$. Explicitly, for $f \in \mathcal{C^\infty(M)}$, $\pi_\alpha(f) = f|_\alpha$.

Moreover, $\{\Ker(|_\alpha)\}_\alpha$ is a covering of $\mathcal{C^\infty(M)}$ since it satisfies \eqref{eq:cov_alg_def}, \textit{i.e.}\ the only smooth function that is zero on every $U_\alpha$ is the zero function. Finally, one can show that $\mathcal{C^\infty(M)} / \big(\Ker(|_\alpha) + \Ker(|_\beta)\big) = \mathcal{C}^\infty(U_\alpha \cap U_\beta)$.

\begin{Remark}
    The canonical map $\pi^\alpha_{\tilde{\beta}}$ used in the above proof introduce another way of going from $\mathcal{A}$ to $\mathcal{A}_{\alpha \tilde{\beta}}$ which is not based on $\pi_{\alpha \tilde{\beta}}$. And symmetrically, one could obtain a third way through the canonical map $\pi^{\tilde{\beta}}_\alpha : \mathcal{A}_{\tilde{\beta}} \to \mathcal{A}_{\alpha\tilde{\beta}}$, so that one get the diagram

    \begin{center}
    \begin{tikzpicture}
		\node (Aa)  at ( 0, 1) {$\mathcal{A}_\alpha$};
		\node (Ab)  at ( 0,-1) {$\mathcal{A}_{\tilde{\beta}}$};
		\node (A)   at (-2, 0) {${}_{\phantom{0}} \mathcal{A}$};
		\node (Aab) at ( 2, 0) {$\mathcal{A}_{\alpha\tilde{\beta}}$};
		\draw[thick, ->] (A.north east) -- 
		    node[anchor=south east]{$\pi_\alpha$} (Aa.west);
		\draw[thick, ->] (A.south east) -- 
		    node[anchor= north east]{$\pi_{\tilde{\beta}}$} (Ab.west);
		\draw[thick, ->] (A.east) --
		    node[anchor=south]{$\pi_{\alpha\tilde{\beta}}$} (Aab.west);
		\draw[thick, ->] (Aa.east) -- 
		    node[anchor=south west]{$\pi^\alpha_{\tilde{\beta}}$} (Aab.north west);
        \draw[thick, ->] (Ab.east) -- 
		    node[anchor=north west]{$\pi_\alpha^{\tilde{\beta}}$} (Aab.south west);
	\end{tikzpicture}%
    \end{center}
    and one can show that $\pi^\alpha_{\tilde{\beta}} \circ \pi_\alpha = \pi^{\tilde{\beta}}_\alpha \circ \pi_{\tilde{\beta}} = \pi_{\alpha\tilde{\beta}}$, and $\pi_\alpha^\alpha = \pi_\alpha$.

    Still, when considering a covering of algebra $\{\mathcal{A}_\alpha\}_{\alpha}$, one can construct $\mathcal{A}_{\alpha\beta}$ through $\mathcal{A}_\alpha$ or through $\mathcal{A}_\beta$. One can consider the covering formed by requiring that these two matches. This covering was called the covering completion of $\mathcal{A}$ in \cite{Calow_1999}. In the $C^*$-algebra case, the covering completion and its algebra are isomorphic, but this is not true in general.
\end{Remark}

\subsection{The star-product quantization with local coordinates}
\paragraph{}
Let $\mathcal{M}$ be a pseudo-Riemannian manifold with an atlas $\{U_\alpha, \varphi^\alpha\}_\alpha$. By definition, the coordinate functions $\varphi^\alpha : U_\alpha \to \Mink$ are smooth diffeomorphisms. In a similar way then section \ref{subsec:gc_q}, one can define the pull-back of $\varphi^\alpha$ as $\varphi^\alpha(f) = f \circ \varphi^\alpha$, for any $f \in \mathcal{C}^\infty(\Mink)$, that is
\begin{align} 
    \varphi^\alpha : \mathcal{C}^\infty(\Mink) \to \mathcal{C}^\infty(U_\alpha),
    \label{eq:qst_pb}
\end{align}
and this pull-back is a homomorphism of algebra.

\paragraph{}
Going into the noncommutative generalisation of this picture implies to consider $(\mathcal{A}_\alpha, \star_\alpha)$ to be the quantum version of the local smooth functions $\mathcal{C}^\infty(U_\alpha)$ and similarly $(\kM, \star_\kappa)$ to be the quantum version of $\mathcal{C}^\infty(\Mink)$. Then, we define a coordinate transformation as an invertible homomorphism 
\begin{align}
    \varphi^\alpha : \kM \to \mathcal{A}_\alpha
    \label{eq:qst_qpb}
\end{align}
where \eqref{eq:qst_qpb} is motivated by \eqref{eq:qst_pb}. The fact that $\varphi^\alpha$ is an invertible homomorphism allows one to relate the star-products of $\mathcal{A}_\alpha$ and $\kM$ through
\begin{align}
    f \star_\alpha g
    = \varphi_\alpha^{-1} \big(\varphi^\alpha(f) \star_\kappa \varphi^\alpha(g) \big),
    \label{eq:qst_loc_sp_rel}
\end{align}
for any $f, g \in \mathcal{A}_\alpha$, similarly as in \eqref{eq:qfst_loc_sp_rel}.

\paragraph{}
Some remarks need to be made at this point. First, as expressed in section \ref{subsec:gc_q}, the quantisation scheme \eqref{eq:qst_loc_sp_rel} explicitly depends on $\varphi^\alpha$, that is on the choice of local coordinates. Another coordinate choice $\tilde{\varphi}^\alpha$ would lead to another star-product with the relation
\begin{align}
    f\ \tilde{\star}_\alpha\ g
    = (\varphi_\alpha^{-1} \circ \tilde{\varphi}^\alpha)^{-1} \Big( (\varphi_\alpha^{-1} \circ \tilde{\varphi}^\alpha)(f) \star_\alpha (\varphi_\alpha^{-1} \circ \tilde{\varphi}^\alpha)(g) \Big)
    \label{eq:qst_eq_sp}
\end{align}
for any $f,g \in \mathcal{A}_\alpha$. The two star-product are therefore related by the change of coordinates $\varphi_\alpha^{-1} \circ \tilde{\varphi}^\alpha : \mathcal{A}_\alpha \to \mathcal{A}_\alpha$. Imposing this change of coordinates to be a symmetry of a physical theory is a direct generalisation of the general covariance principle. In that case, the only remaining degrees of freedom in the quantisation of such a theory corresponds to the choice of quantisation of the Minkowski spacetime (here considered to be the so-called $\kappa$-deformation). Lastly, a similar reasoning can be applied to the involution if we consider $\varphi^\alpha$ to be a homomorphism of $*$-algebras. This allows to write
\begin{align}
    f^* = \varphi_\alpha^{-1} \big( \varphi^\alpha(f)^\dagger \big),
    \label{eq:qst_loc_inv}
\end{align}
for any $f \in \mathcal{A}_\alpha$, and recall that ${}^*$ is the involution of $\mathcal{A}_\alpha$, and ${}^\dagger$ the one of $\kappa$-Minkowski. In the coordinate change invariant scenario, the expression of ${}^*$ does not depend on the choice of coordinates.

\subsection{A quantum version of a locally trivial tangent bundle}
\label{subsec:loc_triv_tang_sp}
\paragraph{}
The principal idea of defining a \emph{locally} trivial bundle is to consider that the bundle is trivial on each $\mathcal{A}_\alpha$. 

\paragraph{}
In the same way as in section \ref{subsec:gctt}, we define a coaction $\coactr: \mathcal{A}_\alpha \to \mathcal{A}_\alpha \otimes \mathfrak{D}_\kappa'$
\begin{align}
    \coactr f
    &= \varphi^{-1}_\alpha \big( \mathfrak{X}_\mu \actl_\kappa \varphi^\alpha(f) \big) \otimes \mathfrak{X}^\mu
    \label{eq:lc_kTcoact}
\end{align}
and its corresponding action $\actl : \kM \otimes \mathcal{A}_\alpha \to \mathcal{A}_\alpha$
\begin{align}
    p_\mu \actl f
    &= \varphi^{-1}_\alpha \big( \mathfrak{X}_\mu \actl_\kappa \varphi^\alpha(f) \big)
    = \varphi^{-1}_\alpha \left( \partial^\mu \varphi^\alpha(f) \right)
    \label{eq:lc_kMact}
\end{align}
for any $f \in \mathcal{A}_\alpha$. This implies that $\mathcal{A}_\alpha$ is both a $\mathcal{T}_\kappa$-comodule algebra and a $\kM$-module algebra.

\paragraph{}
Moreover, we mimic Definition \ref{def:derA_def} to introduce the following terminology.
\begin{definition}
    Let $\{\mathcal{A}_\alpha\}_\alpha$ be a covering of algebra of an associative ${}^*$-algebra $\mathcal{A}$. The algebra $\mathcal{A}$ is said to be related to a locally trivial tangent bundle if $\Der_R(\mathcal{A}_\alpha) = \mathcal{Z(A_\alpha)} \otimes \mathfrak{D}_\kappa$ for any $\alpha$.
    \label{def:loct_tan_bun}
\end{definition}
One can show that $\Der_R(\mathcal{A}_\alpha)$ satisfies the Proposition \ref{prop:derA_prop}. By repeating similar arguments as in section \ref{subsec:triv_diff_calc}, one can also build a full differential structure on $\mathcal{A}_\alpha$.

\paragraph{}
The previous definitions are purely local and seem to very much rely on the covering $\{\mathcal{A}_\alpha\}_\alpha$. In order to be complete, we define below some global notions inherited from the local ones and show that they do not depend on the choice of the covering. This is done through the introduction of the noncommutative generalisation of the partition of unity.

\section{The quantum partition of unity}
\label{sec:ncpu}
\subsection{Smooth manifolds and partition of unity}
\label{subsec:part_unit}
\paragraph{}
In this section, we collect the basic features characterising a partition of unity on a smooth manifold $\mathcal{M}$. There exist different definitions of the notion of partition of unity. We find convenient to start from the following one which serves as a guideline to define a noncommutative extension in the next subsection.

\begin{definition}
    \label{def:part_unit}
    A smooth partition of unity is defined by a set of functions $\{\chi_\alpha\}_\alpha$ satisfying the following conditions: 
    \begin{enumerate}[label = \roman*)]
        \item \label{it:pu_el}
        $\chi_\alpha\in\mathcal{C^\infty(M)}$,
        \item \label{it:pu_supp}
        $\{\Supp(\chi_\alpha)\}_\alpha$ is locally finite,
        \item \label{it:pu_pos}
        $\chi_\alpha \geqslant 0$,
        \item \label{it:pu_id}
        $\sum_\alpha \chi_\alpha = \id_{\mathcal{C^\infty(M)}}$.
    \end{enumerate}
    The partition is said to be subordinate to an open cover $\{U_\alpha\}_\alpha$ of $\mathcal{M}$ if for every $\beta$ there exists $\alpha$ such that $\Supp(\chi_\beta) \subset U_\alpha$.
\end{definition}
Recall that a theorem guaranties that given a smooth paracompact manifold $\mathcal{M}$ with an open cover $\{U_\alpha \}_\alpha$, there always exists a partition of unity subordinate to $\{U_\alpha\}_\alpha$.

\paragraph{}
We use this partition of unity to lift the functions on local $U_\alpha$'s to functions on the full $\mathcal{M}$. To picture how one can do so, let us fix one $\alpha$. Then, define the function $\chi_\alpha : \mathcal{C}^\infty(U_\alpha) \to \mathcal{C}^\infty(U_\alpha)$ by $\chi_\alpha(f_\alpha) = \chi_\alpha f_\alpha$,
for any $f_\alpha \in \mathcal{C}^\infty(U_\alpha)$. The latter map can be extended to
\begin{equation} 
    \chi_\alpha : \mathcal{C}^\infty(U_\alpha) \to \mathcal{C^\infty(M)}.
    \label{eq:loct_pu_dom}
\end{equation}
Indeed, as $\chi_\alpha(x) = 0$ for any $x\ \slashed{\in}\ U_\alpha$, we can write for $x \in \mathcal{M}\backslash U_\alpha$, $\chi_\alpha(x)f _\alpha(x) = 0$ even if $f_\alpha$ is not \textit{a priori} defined outside $U_\alpha$. Thus, the partition of unity allows to lift functions from $\mathcal{C}^\infty(U_\alpha)$ to $\mathcal{C^\infty(M)}$ which thus establishes a bridge between local and global quantities.

\paragraph{}
Denoting by $f|_\alpha \in \mathcal{C}^\infty(U_\alpha)$, the restriction of $f \in \mathcal{C^\infty(M)}$ to $U_\alpha$, one can write
\begin{align}
    f = \sum_\alpha \chi_\alpha  f|_\alpha .
    \label{eq:loc2glob}
\end{align}
This relation may be a trivial property of the partition of unity but is at the very basis of what we want to accomplish. Here, in \eqref{eq:loc2glob} the local function $f|_\alpha$ is build from cutting off the global function $f$ through the restriction $|_\alpha$. In the following, we consider similar global objects built from local ones, thanks to the partition of unity $\chi_\alpha$.

\paragraph{}
Other useful properties, needed in the following, are listed below without proof.

\paragraph{}
First, a product of partition of unity is itself a partition of unity. Indeed, consider $\{U_\alpha\}_\alpha$ and $\{\widetilde{U}_{\tilde{\beta}}\}_{\tilde{\beta}}$ be complete sets of charts of $\mathcal{M}$. Let $\{\chi_\alpha\}_\alpha$ and $\{\tilde{\chi}_{\tilde{\beta}}\}_{\tilde{\beta}}$ be partitions of unity subordinate to $\{U_\alpha\}_\alpha$ and $\{\widetilde{U}_{\tilde{\beta}}\}_{\tilde{\beta}}$ respectively. Then, $\{\chi_\alpha \tilde{\chi}_{\tilde{\beta}}\}_{\alpha,\tilde{\beta}}$ is a partition of unity subordinate to $\{U_\alpha \cap \widetilde{U}_{\tilde{\beta}}\}_{\alpha, \tilde{\beta}}$. Here $(\chi_\alpha \tilde{\chi}_{\tilde{\beta}})(x) = \chi_\alpha(x) \tilde{\chi}_{\tilde{\beta}}(x)$ correspond to the point wise product of the partitions.

When considering functionals $\chi$, this statement becomes that $\chi_\alpha \circ \tilde{\chi}_{\tilde{\beta}} = \tilde{\chi}_{\tilde{\beta}} \circ \chi_\alpha$ is a partition of unity subordinate to $\{U_\alpha \cap \widetilde{U}_{\tilde{\beta}}\}_{\alpha, \tilde{\beta}}$.

\paragraph{}
Next, the positivity condition \ref{it:pu_pos}, $\chi_\alpha \geqslant 0$, in Definition \ref{def:part_unit}, can be replaced, for the functional $\chi$, by the statement that if for any $f_\alpha \in \mathcal{C}^\infty(U_\alpha)$, such that $f_\alpha \geqslant 0$, then
\begin{equation}
    \chi_\alpha(f_\alpha) \geqslant 0.
    \label{comm-positiv}
\end{equation}  

\paragraph{}
Finally, recall that, given an atlas $\{(U_\alpha,\varphi^\alpha)\}_\alpha$ of $\mathcal{M}$, the map $\partial_\mu^\alpha : U_\alpha \to TU_\alpha$, which associates to any point of $U_\alpha$ the derivative of the $\mu$-th coordinate given by $\varphi^\alpha$, is a (local) vector field and a (local) frame of $U_\alpha$ in $T\mathcal{M}$. In other words, using local triviality 
\begin{equation}
    TU_\alpha \simeq U_\alpha \times \mathbb{R}^{1,d}, 
    \label{eq:loc_triv}
\end{equation}
$\{\partial_\mu^\alpha\}_\mu$ is a basis vector of $\mathbb{R}^{1,d}$. For any vector field $X \in \Gamma(\mathcal{M})$, one can write\footnote{
Recall that summation convention holds for $\mu$ but not for $\alpha$ here.}  $X|_{U_\alpha} = X^\mu_\alpha \partial_\mu^\alpha$, with $X^\mu_\alpha\in\mathcal{C}^\infty(U_\alpha)$. And, as discussed above, given a partition of unity $\{\chi_\alpha\}_\alpha$ subordinate to $\{U_\alpha\}_\alpha$, one has
\begin{align}
	X 
	= \sum_\alpha X^\mu_\alpha \chi_\alpha \partial^\alpha_\mu
	= \sum_\alpha \chi_\alpha(X^\mu_\alpha \partial^\alpha_\mu),
	\label{eq:commu_der_comp}
\end{align} 
(see for example \cite{Rudolph_2012} Remark 2.3.11.1). In a similar way to \eqref{eq:loc2glob}, this formula links the local vector fields $X^\mu_\alpha \partial^\alpha_\mu$, that express simply in the $\{\partial_\mu^\alpha\}_\mu$ basis, to the global vector field $X$.

\subsection{The noncommutative partition of unity}
\label{subsec:noncopar-comod}
\paragraph{}
We introduce here the definition of a noncommutative partition of unity, first on its own and then adapted to a covering of algebra.

\subsubsection{Definition and properties}
\label{subsubsec:loct_nc_part_unit}
\paragraph{}
Let $\Phi_\mathcal{A}$ be the set of characters of $\mathcal{A}$, that is of all homomorphisms from $\mathcal{A}$ to $\mathbb{C}$. For an element $f \in \mathcal{A}$, we denote $\Supp(f) = \{ \phi \in \Phi_\mathcal{A},\ \phi(f) \neq 0\}$. We now generalise the Definition \ref{def:part_unit} as 
\begin{definition}
    \label{def:nc_part_unit}
    Let $\mathcal{A}$ be a ${}^*$-algebra, a noncommutative partition of unity is a set of elements $\{\chi_\alpha\}_\alpha$ satisfying
    \begin{enumerate}[label = \roman*)]
        \item \label{it:ncpu_el}
        $\chi_\alpha \in \mathcal{A}$.
        \item \label{it:ncpu_supp}
        The set $\{\Supp(\chi_\alpha)\}_\alpha$ is locally finite.
        \item \label{it:ncpu_pos}
        $\chi_\alpha \geqslant 0$, \textit{i.e.} it exists $\zeta_\alpha \in \mathcal{A}$ such that $\chi_\alpha = \zeta_\alpha \zeta_\alpha^*$.
        \item \label{it:ncpu_id}
        For any $f \in \mathcal{A}$, $\sum_\alpha \chi_\alpha \, f = f$.
    \end{enumerate}
\end{definition}

\paragraph{}
It is useful to comment the extension of the conditions \ref{it:pu_el} - \ref{it:pu_id} of Definition \ref{def:part_unit} to their respective noncommutative counterparts  \ref{it:ncpu_el} - \ref{it:ncpu_id} in Definition \ref{def:nc_part_unit}.

\paragraph{}
First, keeping in mind the Gelfand-Naimark theorem together with the Gelfand transform holding for commutative ($C^*$-)algebras, one is naturally led to trade the notion of point by the notion of character of the algebra. Therefore, a reasonable noncommutative generalisation of the notion of support of a smooth function uses the set of characters of the algebra, as defined above.

Then, the local finitude of the support of the partition of unity (corresponding to \ref{it:pu_supp} in Definition \ref{def:part_unit}) can be naturally extended by the requirement \ref{it:ncpu_supp} above. The latter can be also stated as follows: for any $\phi \in \Phi_\mathcal{A}$, there exists a neighbourhood $V$ of $\phi$ such that for any $\tilde{\phi} \in V$, we have $\tilde{\phi}(\chi_\alpha) \neq 0$ for a finite number of index $\alpha$. The previous notion of neighbourhood needs a topology, which can be taken to be the weak* topology.

Next, the extension of the positivity condition is given by \ref{it:ncpu_pos} which reduces to the usual positivity condition whenever $\mathcal{A} = \mathcal{C^\infty(M)}$.

Finally, the sum of all $\chi_\alpha$ must leave the base space invariant, as expressed in \ref{it:pu_id} of Definition \ref{def:part_unit}.

\begin{proposition}
    \label{prop:ncpu_prod}
    Let $\mathcal{A}$ be a ${}^*$-algebra with two partitions of unity $\{\chi_\alpha\}_\alpha$ and $\{\tilde{\chi}_{\tilde{\beta}}\}_{\tilde{\beta}}$. Then, $\{\chi_\alpha \bullet \tilde{\chi}_{\tilde{\beta}}\}_{\alpha, \tilde{\beta}}$ is a partition of unity of $\mathcal{A}$, where the action $\bullet$ is defined as $\chi_\alpha \bullet f = \zeta_\alpha\, f\, \zeta_\alpha^*$, for any $f \in \mathcal{A}$ and with $\zeta_\alpha$ defined in \ref{it:ncpu_pos}.
\end{proposition}
\begin{proof}
    First, one has $\chi_\alpha \bullet \tilde{\chi}_{\tilde{\beta}} \in \mathcal{A}$ as a product of elements of $\mathcal{A}$.

    Second, let $V$ be the neighbourhood of any character of $\mathcal{A}$. Let $\phi \in V$, then exists finite number of $\alpha$ and $\tilde{\beta}$ such that $\phi(\chi_\alpha) \neq 0$ and $\phi(\tilde{\chi}_{\tilde{\beta}}) \neq 0$. Still, as $\phi(\chi_\alpha \bullet \tilde{\chi}_{\tilde{\beta}}) = \phi(\chi_\alpha) \phi(\tilde{\chi}_{\tilde{\beta}})$, one has that $\phi(\chi_\alpha \bullet \tilde{\chi}_{\tilde{\beta}}) \neq 0$ if and only if $\phi(\chi_\alpha) \neq 0$ and $\phi(\tilde{\chi}_{\tilde{\beta}}) \neq 0$, so that only a finite number of $\alpha$ and $\tilde{\beta}$ satisfies this condition.

    Third, it exists $\zeta_\alpha, \tilde{\zeta}_{\tilde{\beta}} \in \mathcal{A}$ such that $\chi_\alpha = \zeta_\alpha \zeta_\alpha^*$ and $\tilde{\chi}_{\tilde{\beta}} = \tilde{\zeta}_{\tilde{\beta}} \tilde{\zeta}_{\tilde{\beta}}^*$. Thus, one can write $\chi_\alpha \bullet \tilde{\chi}_{\tilde{\beta}} = \zeta_\alpha \tilde{\zeta}_{\tilde{\beta}} (\zeta_\alpha \tilde{\zeta}_{\tilde{\beta}})^*$.

    Finally, for $f \in \mathcal{A}$, one has
    \begin{align*}
        \sum_{\alpha, \tilde{\beta}} (\chi_\alpha \bullet \tilde{\chi}_{\tilde{\beta}}) f
        &= \sum_{\alpha, \tilde{\beta}} \zeta_\alpha \tilde{\zeta}_{\tilde{\beta}} (\zeta_\alpha \tilde{\zeta}_{\tilde{\beta}})^* \, f
        = \sum_{\alpha} \zeta_\alpha \left( \sum_{\tilde{\beta}} \tilde{\zeta}_{\tilde{\beta}} \tilde{\zeta}_{\tilde{\beta}}^* \right) \zeta_\alpha^* \, f
        = \sum_{\alpha} \zeta_\alpha \zeta_\alpha^* \, f
        = f.
    \end{align*}
\end{proof}

\paragraph{}
We now introduce examples of partition of unity on noncommutative spaces. A trivial example in the case of a unital algebra is the set $\{1\}$, which defines a partition of unity. Still, more relevant examples are the matrix algebra and the Moyal space. In the Moyal space, we introduce the so-called matrix basis which makes it close to a matrix algebra. Still, the condition \ref{it:ncpu_supp} is different in the two cases because the matrix basis of Moyal is of infinite dimension.

\begin{Example}\label{ex:matrix}
    Let $\mathcal{A} = M_N(\mathbb{C})$, the algebra of square complex matrices of size $N$. We introduce its canonical basis $\{E_{mn}\}_{m,n \in \{1, \dots, N\}}$, where $E_{mn}$ correspond to the matrix having all zero elements, except the element of line $m$ and column $n$ which is $1$. This basis satisfy $E_{mn} E_{kl} = \delta_{nk} E_{ml}$.

    We introduce the set of elements $\{\chi_{m}\}_{m \in \{1, \dots, N\}}$, with $\chi_m = E_{mm}$. This corresponds to a partition of unity on the algebra of matrices. Indeed, \ref{it:ncpu_el} is satisfied as $E_{mm} \in M_N(\mathbb{C})$. The condition \ref{it:ncpu_supp} is also trivially fulfilled since this partition of unity is finite. One can write $\chi_{m} = E_{mm} E_{mm}$ so that we have \ref{it:ncpu_pos}. Finally, for any $M \in M_N(\mathbb{C})$, it can be decomposed in the canonical base as $M = \sum_{k,l} M_{kl} E_{kl}$, and so
    \begin{equation*}
        \sum_m \chi_m M
        = \sum_{m,k,l} M_{kl} E_{mm} E_{kl}
        = \sum_{m,k,l} M_{kl} \delta_{mk} E_{ml}
        = \sum_{ml} M_{ml} E_{ml}
        = M,
    \end{equation*}
    which corresponds to \ref{it:ncpu_id}.
\end{Example}

\begin{Example}\label{ex:Moyal}
    We consider the Moyal space $\mathcal{A} = \Moyal$ and, for convenience, we introduce the matrix basis $\{f_{mn}\}_{m,n \in \mathbb{N}}$ \cite{graciavar1}. The latter is an orthonormal basis of $\mathcal{S}(\mathbb{R}^{2d})$, the set of Schwartz function of $\mathbb{R}^{2d}$, in the following sense
    \begin{subequations}
    \begin{align}
        f_{mn} \star_\theta f_{kl} &= \delta_{nk} f_{ml}, \qquad
        f_{mn}^\dagger = f_{nm}, 
        \label{eq:moyal_mb_prod} \\
        \langle f_{mn}, f_{kl} \rangle_{L^2(\mathbb{R}^{2d})}
        &= \int \tdl{2d}{x} f_{mn}^\dagger \star_\theta f_{kl} (x)
        = 2\pi \theta \delta_{mk}\delta_{nl}
    \end{align}
    \end{subequations}
    where $\star_\theta$ denotes the star-product of Moyal and ${}^\dagger$ its involution. Then, any element $a \in \Moyal$ and any linear function $\phi : \Moyal \to \mathbb{C}$ can be decomposed as
    \begin{align}
        a &= \sum_{m,n} a_{mn} f_{mn}, &
        \phi &= \sum_{m,n} \phi_{mn} \langle f_{mn}, \cdot \rangle_{L^2(\mathbb{R}^{2d})},
        \label{eq:moyal_mb_decomp}
    \end{align}
    where $a_{mn}, \phi_{mn} \in \mathbb{C}$.
    
    We introduce the set of elements $\{\chi_{m}\}_{m\in\mathbb{N}}$ with $\chi_m = f_{mm}$. The condition \ref{it:ncpu_el} is fulfilled since $f_{mm} \in \Moyal$, and from \eqref{eq:moyal_mb_prod}, one has $\chi_m = f_{mm} f_{mm}$ so that \ref{it:ncpu_pos} is also satisfied. The requirement \ref{it:ncpu_id} is similar to the finite matrices case. Let $a \in \Moyal$, then we decompose it as in \eqref{eq:moyal_mb_decomp}, so that
    \begin{equation*}
        \sum_m \chi_m \star_\theta a
        = \sum_{m,k,l} a_{kl} f_{mm} \star_\theta f_{kl}
        = \sum_{m,k,l} a_{kl} \delta_{mk} f_{ml}
        = \sum_{ml} a_{ml} f_{ml}
        = a.
    \end{equation*}
    Now, let $\phi \in \Phi_{\Moyal}$ and
    \begin{align*} 
        V =
        \left\{ \tilde{\phi} \in \Phi_{\Moyal},\ \underset{a \in \Moyal}{\sup} | \phi(a) - \tilde{\phi}(a) | < 1,\ \exists N\in\mathbb{N},\ \tilde{\phi}_{mn} = 0 \text{ if } m+n > N \right\}.
    \end{align*}
    $V$ corresponds to the set of characters of $\Moyal$ that are inside the unit ring of $\phi$ and such that their decomposition \eqref{eq:moyal_mb_decomp} is zero from a certain rank. Therefore, for $\tilde{\phi} \in V$, one can compute $\tilde{\phi}(\chi_m) = 2\pi \theta \tilde{\phi}_{mm}$, which is zero for $2m > N$. This proves that the condition \ref{it:ncpu_supp} holds.
\end{Example}

\paragraph{}
One could question the fact that we considered the left product by $\chi_\alpha$ in \ref{it:ncpu_id} of Definition \ref{def:nc_part_unit}. It appears that choosing the right multiplication does not change the previous results and would amount to the same thing in the commutative case. Still, if we change \ref{it:ncpu_id} to $\sum_\alpha \chi_\alpha \bullet f = f$, with $\bullet$ defined in Proposition \ref{prop:ncpu_prod}, one can show that there is only one partition of unity (up to positive scalar multiplication) in Examples \ref{ex:matrix} and \ref{ex:Moyal} which is the singlet containing the identity matrix (or correspondingly $\chi = \sum_m f_{mm}$ in the case of Moyal), and so even if $\bullet$ amount to the left/right multiplication in the commutative case.

\subsubsection{The adapted partition of unity}
\label{subsubsec:loct_adap_ncpu}
\paragraph{}
Inspired by the commutative setting, one defines a subordinate partition of unity as
\begin{definition}
    \label{def:partition}
    Let $\{\mathcal{A}_\alpha\}_\alpha$ be a covering of algebra of an associative ${}^*$-algebra $\mathcal{A}$ (as defined in section \ref{subsec:loct_ideal}). A noncommutative partition of unity subordinate to the covering of algebra $\{\mathcal{A}_\alpha\}_\alpha$ of $\mathcal{A}$ is a partition of unity $\{\chi_\beta\}_\beta$ if for every $\beta$ there exists $\alpha$ such that
    \begin{align}
        \Supp(\chi_\beta)
        \subset \Ker(J_\alpha)
        = \big\{ \phi \in \Phi_\mathcal{A}, \ \phi(f) = 0 \ \text{for any } f \in J_\alpha \big\},
        \label{eq:ncpu_sub_def}
    \end{align}
    where $J_\alpha$ is the ideal defining the covering. We further say that the partition $\{\chi_\alpha\}_\alpha$ is adapted to the covering $\{\mathcal{A}_\alpha\}_\alpha$ if $\beta = \alpha$ in \eqref{eq:ncpu_sub_def}.
\end{definition}

In the commutative case, one has that $J_\alpha = \Ker(\vert_\alpha)$ so that $\Ker(J_\alpha) = \{ x \in \mathcal{M},\ f(x) = 0,\ \forall f \in \Ker(J_\alpha) \} = U_\alpha$. Therefore, the subordinate condition perfectly match the one of Definition \ref{def:part_unit}. The adapted condition \eqref{eq:ncpu_sub_def} boils down to $\Supp(\chi_\alpha) \subset U_\alpha$, as $\Ker(J_\alpha) = U_\alpha$ in the commutative setting.

\paragraph{}
As introduced in section \ref{subsec:part_unit}, it is more convenient to work with the functional version of the partition of unity $\chi_\alpha$. In this sense, it can also be defined in the noncommutative case as a functional $\chi_\alpha : \mathcal{A}_\alpha \to \mathcal{A}$, given for any $f \in \mathcal{A}$ by  
\begin{equation}
  \chi_\alpha \circ \pi_\alpha (f) = \chi_\alpha f.
  \label{eq:ncpu_func_def}
\end{equation}

\paragraph{}
The fact that $\chi_\alpha$ starts in the local algebra $\mathcal{A}_\alpha$ and ends in the global one $\mathcal{A}$, namely $\chi_\alpha : \mathcal{A}_\alpha \to \mathcal{A}$, is the noncommutative analogue of \eqref{eq:loct_pu_dom}. As opposed to the linearity of $\chi_\alpha$ which arises as a property (see Proposition \ref{prop:part_unit_comp} below). Finally one should note that for any $f \in \mathcal{A}$,
\begin{align}
    \sum_\alpha \chi_\alpha \circ \pi_\alpha (f)
    = \sum_\alpha \chi_\alpha f
    = f
    \label{eq:ncpu_loc2glob}
\end{align}
which is the noncommutative counterpart of \eqref{eq:loc2glob}.

\paragraph{}
Some properties of such a partition of unity are discussed below.

\paragraph{}
First, one observes that compatibility of the partition of unity with the algebraic structures of $\mathcal{A}$ yields the following property.

\begin{proposition}\label{prop:part_unit_comp}
    For any $f, g \in \mathcal{A}$, for any partition of unity $\{\chi_\alpha\}_\alpha$ adapted to the covering $\{\mathcal{A}_\alpha\}_\alpha$, one has for any $\alpha$
    \begin{enumerate}[label = \roman*)]
        \item \label{it:part_unit_lin}
        $\chi_\alpha$ is linear.
        \item \label{it:part_unit_cent}
        $\chi_\alpha \circ \pi_\alpha \in \mathcal{Z}\big(\Lin(\mathcal{A})\big).$
        \item \label{it:part_unit_prod}
        $\chi_\alpha$ is a morphism of $\mathcal{A}$-modules, where the action of $\mathcal{A}_\alpha$ on $\mathcal{A}$ given by $\pi_\alpha(fg) = \pi_\alpha(f) g$, and the action of $\mathcal{A}$ on itself is given by its product. Explicitly, for any $f, g \in \mathcal{A}$,
        \begin{align}
            \chi_\alpha ( \pi_\alpha(f) ) g
            = \chi_\alpha(\pi_\alpha(f) g).
            \label{eq:part_unit_comp_prod}
        \end{align}
        \item \label{it:part_unit_commu}
        For any other partition of unity $\{\chi_{\tilde{\beta}}\}_{\tilde{\beta}}$ adapted to a covering $\{\mathcal{A}_{\tilde{\beta}}\}_{\tilde{\beta}}$, one has, for any $\alpha, \tilde{\beta}$,
        \begin{align}
            \chi_\alpha \chi_{\tilde{\beta}} 
            = \chi_{\tilde{\beta}} \chi_\alpha 
            = \chi_\alpha \bullet \chi_{\tilde{\beta}} 
            = \chi_{\tilde{\beta}} \bullet \chi_\alpha.
            \label{eq:part_unit_commu}
        \end{align}
    \end{enumerate}
\end{proposition}

\begin{proof}
    First, lets prove \ref{it:part_unit_lin}. For any $\lambda, \mu \in \mathbb{C}$, and $f, g \in \mathcal{A}$, by considering the expression $\lambda f + \mu g$, one can choose to apply \eqref{eq:ncpu_loc2glob} on either $f$ and $g$ or to the element $\lambda f + \mu g$. Thus,
    \begin{align*}
        \lambda f + \mu g
        &= \sum_\alpha \chi_\alpha \circ \pi_\alpha(\lambda f + \mu g)
        = \sum_\alpha \chi_\alpha ( \lambda \pi_\alpha(f) + \mu \pi_\alpha(g) ) \\
        &= \sum_\alpha \lambda \chi_\alpha \circ \pi_\alpha(f) + \mu \chi_\alpha \circ \pi_\alpha(g)
    \end{align*}
    where we used the linearity of the property of the canonical projection $\pi_\alpha(\lambda f + \mu g) =  \lambda \pi_\alpha(f) + \mu \pi_\alpha(g)$. Then, for a given $\alpha_0$, one can set $f, g \in \bigcup \limits_{\alpha \neq \alpha_0} J_{\alpha}$ so that $\sum_\alpha \lambda \chi_\alpha \circ  \pi_\alpha(f) + \mu \chi_\alpha \circ \pi_\alpha(g) = \lambda \chi_{\alpha_0} \circ \pi_{\alpha_0}(f) + \mu \chi_{\alpha_0} \circ \pi_{\alpha_0}(g)$ and similarly $\sum_\alpha \chi_\alpha ( \lambda \pi_\alpha(f) + \mu \pi_\alpha(g) ) = \chi_{\alpha_0} ( \lambda \pi_{\alpha_0}(f) + \mu \pi_{\alpha_0}(g) )$.

    One can also use \eqref{eq:ncpu_loc2glob} to prove \ref{it:part_unit_cent}. Explicitly, let $F \in \Lin(\mathcal{A})$, then for any $f \in \mathcal{A}$,
    \begin{align*}
        F(f)
        &= \sum_\alpha \chi_\alpha \circ \pi_\alpha (F(f)) \\
        &= F\left( \sum_\alpha \chi_\alpha \circ \pi_\alpha (f) \right)
        = \sum_\alpha F \circ \chi_\alpha \circ \pi_\alpha (f)
    \end{align*}
    and given an $\alpha_0$, one can set $f \in \bigcup \limits_{\alpha \neq \alpha_0} J_\alpha$, so that $\chi_{\alpha_0} \circ \pi_{\alpha_0} \circ F = F \circ \chi_{\alpha_0} \circ \pi_{\alpha_0}$.

    Similarly, let us prove \ref{it:part_unit_prod}. By considering the product $ab \in \mathcal{A}$, one can choose to apply \eqref{eq:ncpu_loc2glob} on either $f$ or the full product $fg$, explicitly
    \begin{align*}
        fg
        &= \sum_\alpha \chi_\alpha \circ \pi_\alpha (fg)
        = \sum_\alpha \chi_\alpha( \pi_\alpha(f) g) \\
        &= \sum_\alpha \chi_\alpha \circ \pi_\alpha(f) g
    \end{align*}
    and given an $\alpha_0$, one can set $f \in \bigcup \limits_{\alpha \neq \alpha_0} J_{\alpha}$ so that $\sum_\alpha \chi_\alpha \circ \pi_\alpha(f) g = \chi_{\alpha_0}\pi_{\alpha_0}(f) g$ and $\sum_\alpha \chi_\alpha(\pi_\alpha(f) g) = \chi_{\alpha_0}(\pi_{\alpha_0}(f) g)$.

    Finally, the proof of \ref{it:part_unit_commu} is similar to the one above by simply considering a element $a \in \mathcal{A}$ and using \eqref{eq:ncpu_loc2glob} to decompose it first in $\alpha$ and in $\tilde{\beta}$, or the converse, or both at the same time. Explicitly,
    \begin{align*}
        f
        &= \sum_\alpha \chi_\alpha f
        = \sum_{\alpha, \tilde{\beta}} \chi_{\tilde{\beta}} \chi_\alpha f \\
        &= \sum_{\tilde{\beta}} \chi_{\tilde{\beta}} f
        = \sum_{\alpha, \tilde{\beta}} \chi_\alpha \chi_{\tilde{\beta}} f \\
        &= \sum_{\alpha, \tilde{\beta}} \chi_\alpha \bullet \chi_{\tilde{\beta}} f \\
        &= \sum_{\alpha, \tilde{\beta}} \chi_{\tilde{\beta}} \bullet \chi_\alpha f,
    \end{align*}
    and, as previously, considering $\alpha_0, \tilde{\beta}_0$, one can choose $f \in \bigcup \limits_{\alpha \neq \alpha_0, \tilde{\beta} \neq \tilde{\beta}_0} J_{\alpha} + J_{\tilde{\beta}}$. This proves \eqref{eq:part_unit_commu} for any $\alpha_0, \tilde{\beta}_0$.
\end{proof}

\paragraph{}
Now considering two partitions of unity adapted to two different coverings of algebra, we know from Proposition \ref{prop:ncpu_prod} that their product is a partition of unity. The question whether this product of partition is itself adapted to some covering is answered in the following proposition
\begin{proposition}
    \label{prop:ncpu_prod_adapt}
    Let $\mathcal{A}$ be a ${}^*$-algebra with two partitions of unity $\{\chi_\alpha\}_\alpha$ and $\{\tilde{\chi}_{\tilde{\beta}}\}_{\tilde{\beta}}$, adapted to the coverings of algebra $\{\mathcal{A}_\alpha\}_{\alpha}$ and $\{\mathcal{A}_{\tilde{\beta}}\}_{\tilde{\beta}}$ respectively. Then, the partition of unity $\{\chi_\alpha \bullet \tilde{\chi}_{\tilde{\beta}}\}_{\alpha,\tilde{\beta}}$ is adapted to $\{\mathcal{A}_{\alpha\tilde{\beta}}\}_{\alpha, \tilde{\beta}}$.
\end{proposition}

\begin{proof}
    Let us consider $\chi_\alpha \bullet \tilde{\chi}_{\tilde{\beta}}$ for some fixed $\alpha$ and $\tilde{\beta}$, together with $\phi \in \Supp(\chi_\alpha \bullet \tilde{\chi}_{\tilde{\beta}})$. By definition, $\phi(\chi_\alpha \bullet \tilde{\chi}_{\tilde{\beta}}) \neq 0$, which upon using \eqref{eq:part_unit_commu} writes $\phi(\chi_\alpha) \phi(\tilde{\chi}_{\tilde{\beta}}) \neq 0$. Therefore, one has that $\phi \in \Supp(\chi_\alpha) \cap \Supp(\tilde{\chi}_{\tilde{\beta}})$, but as the partitions are subordinate $\Supp(\chi_\alpha) \cap \Supp(\tilde{\chi}_{\tilde{\beta}}) \subset \Ker(J_\alpha) \cap \Ker(\tilde{J}_{\tilde{\beta}}) = \Ker(J_\alpha \cup \tilde{J}_{\tilde{\beta}})$. But, by definition $J_\alpha + \tilde{J}_{\tilde{\beta}}$ contains $J_\alpha$ and $\tilde{J}_{\tilde{\beta}}$, so that $\Ker(J_\alpha \cup \tilde{J}_{\tilde{\beta}}) \subset \Ker(J_\alpha + \tilde{J}_{\tilde{\beta}})$, hence $\phi \in \Ker(J_\alpha + \tilde{J}_{\tilde{\beta}})$.
\end{proof}

\paragraph{}
Finally, note that the commutative version of Proposition \ref{prop:ncpu_prod_adapt} states that given two partitions of unity $\{\chi_\alpha\}_\alpha$ and $\{\tilde{\chi}_{\tilde{\beta}}\}_{\tilde{\beta}}$ adapted to two open covers $\{U_\alpha\}_\alpha$ and $\{\widetilde{U}_{\tilde{\beta}}\}_{\tilde{\beta}}$ respectively, then $\{\chi_\alpha\tilde{\chi}_{\tilde{\beta}}\}_{\alpha,\tilde{\beta}}$ is a partition of unity adapted to the covering $\{U_\alpha \cap \widetilde{U}_{\tilde{\beta}}\}_{\alpha,\tilde{\beta}}$.

\paragraph{}
As a side remark, we point out that the Propositions \ref{prop:ncpu_prod} and \ref{prop:ncpu_prod_adapt} can be exported to the functional formalism of the partition of unity. Indeed, using \eqref{eq:ncpu_func_def} and \eqref{eq:part_unit_commu}, and denoting $ \chi_{\alpha\tilde{\beta}} = \chi_\alpha \tilde{\chi}_{\tilde{\beta}}$, one can show that
\begin{align}
    \chi_{\alpha\tilde{\beta}} \circ \pi_{\alpha\tilde{\beta}}
    = \tilde{\chi}_{\tilde{\beta}} \circ \pi_{\tilde{\beta}} \circ \chi_\alpha \circ \pi_\alpha.
    \label{eq:ncpuf_prod}
\end{align}

\subsection{Global vector fields and forms}
\label{subsec:loc_triv_diff_calc}
\paragraph{}
The focus is now to use the partition of unity to link the (restricted) derivations of $\mathcal{A}_\alpha$ to the ones of $\mathcal{A}$. Indeed, the partition of unity defined in section \ref{subsubsec:loct_adap_ncpu} can either be used to project derivations of $\mathcal{A}$ on a local algebra $\mathcal{A}_\alpha$, or to ``lift derivation'' from $\mathcal{A}_\alpha$ to $\mathcal{A}$. The following proposition holds true.

\begin{proposition}
    \label{prop:der_loc2glob}
    Let $X \in \Der(\mathcal{A})$, then
    \begin{align} 
        X_\alpha
        = \pi_\alpha \circ X \circ \chi_\alpha
        \in \Der(\mathcal{A}_\alpha).
        \label{eq:ncder_loc2glob}
    \end{align}
    Correspondingly, let $X_\alpha \in \Der(\mathcal{A}_\alpha)$, then
    \begin{align}
        X
        = \chi_\alpha \circ X_\alpha \circ \pi_\alpha
        \in \Der(\mathcal{A}).
        \label{eq:ncder_glob2loc}
    \end{align}
\end{proposition}

\begin{proof}
    For given $X_\alpha \in \Der(\mathcal{A}_\alpha)$, let $X = \chi_\alpha \circ X_\alpha \circ \pi_\alpha$. One computes for any $f, g \in \mathcal{A}$,
	\begin{align*}
		X (fg) 
		&= \chi_\alpha \circ X_\alpha(\pi_\alpha(f) \pi_\alpha(g)) \\
		&= \chi_\alpha \big( X_\alpha (\pi_\alpha(f)) \pi_\alpha(g) + \pi_\alpha(f) X_\alpha(\pi_\alpha(g)) \big) \\
        &= \chi_\alpha \Big( X_\alpha (\pi_\alpha(f)) g + f X_\alpha(\pi_\alpha(g)) \Big) \\
		&= \chi_\alpha \circ X_\alpha \circ \pi_\alpha (f) g + f \chi_\alpha \circ X_\alpha \circ \pi_\alpha (g) \\
		&= X(f) g + f X(g).
	\end{align*}
    Here, on the third line, we used the fact that for any $f_\alpha \in \mathcal{A}_\alpha$ and for any $g \in \mathcal{A}$, $f_\alpha g = f_\alpha \pi_\alpha(g)$ by definition of the canonical action of $\mathcal{A}$ on $\mathcal{A}_\alpha$.

    Now we turn to \eqref{eq:ncder_glob2loc}. Let $X \in \Der(\mathcal{A})$ and define $X_\alpha = \pi_\alpha \circ X \circ \chi_\alpha$, then for any $f, g \in \mathcal{A}$,
    \begin{align*}
        X_\alpha(\pi_\alpha(f) \pi_\alpha(g))
        &= \pi_\alpha \circ X \circ \chi_\alpha \circ \pi_\alpha (f g) \\
        &= \pi_\alpha \circ \chi_\alpha \circ \pi_\alpha \circ X(f g) \\
        &= \pi_\alpha \circ \chi_\alpha \circ \pi_\alpha \big( X(f) g + f X(g) \big) \\
        &= \pi_\alpha \big( \chi_\alpha \circ \pi_\alpha \circ X(f) g + f \chi_\alpha \circ \pi_\alpha \circ X(g) \big) \\
        &= X_\alpha \circ \pi_\alpha(f) g + f X_\alpha \circ \pi_\alpha(g)
        = X_\alpha \circ \pi_\alpha(f) \pi_\alpha(g) + \pi_\alpha(f) X_\alpha \circ \pi_\alpha(g).
    \end{align*}
    Here, we used \ref{it:part_unit_lin} of Proposition \ref{prop:part_unit_comp} on the second line and \eqref{eq:part_unit_comp_prod} together with the canonical action of $\mathcal{A}$ on $\mathcal{A}_\alpha$ several times. Note that the previous computation is valid since $\pi_\alpha$ is surjective.
\end{proof}

\paragraph{}
From now on, the restricted derivations of $\mathcal{A}$ are defined as
\begin{align}
    \Der_R(\mathcal{A})
    = \left\{ \sum_\alpha \chi_\alpha \circ X_\alpha \circ \pi_\alpha, \; X_\alpha \in \Der_R(\mathcal{A}_\alpha) \right\}.
    \label{eq:loct_fullderR_def}
\end{align}
Note that, as discussed in section \ref{subsec:loc_triv_tang_sp}, this set of derivations may not be the full set of derivations of $\mathcal{A}$. Still, these derivations can be put under the form \eqref{eq:commu_der_comp}, as given by the following
\begin{proposition}
    Any $X \in \Der_R(\mathcal{A})$ can be expressed as 
    \begin{align} 
        X = \sum_\alpha \chi_\alpha(X^\mu_\alpha p^\alpha_\mu),
        \label{eq:loc_triv_der_comp}
    \end{align} 
    where $X_\alpha^\mu \in\mathcal{Z(A_\alpha)}$ and $p^\alpha_\mu = p_\mu \circ \pi_\alpha$, with $p_\mu$ a basis element of $\kM$.
\end{proposition}

\begin{proof}
    Since $\Der_R(\mathcal{A}_\alpha) = \mathcal{Z(A_\alpha)} \otimes \mathfrak{D}_\kappa$, then for $X_\alpha \in \Der_R(\mathcal{A}_\alpha)$, one can find $X_\alpha^\mu \in\mathcal{Z(A_\alpha)}$ such that $X_\alpha = X^\mu_\alpha p_\mu$, with $\{p_\mu\}_\mu$ being the generators of $\kM$. Then, for any $f \in \mathcal{A}$
    \begin{align*}
        X(f)
        &= \sum_\alpha \chi_\alpha \circ X_\alpha \circ \pi_\alpha (f)
        = \sum_\alpha \chi_\alpha \Big( X^\mu_\alpha (p_\mu \actl \pi_\alpha(f)) \Big)
        = \sum_\alpha \chi_\alpha (X^\mu_\alpha p^\alpha_\mu (f)),
    \end{align*}
    where the action $\actl$ of $\kM$ on $\mathcal{A}_\alpha$ is defined in \eqref{eq:modalg_loc_act}.
\end{proof}

\paragraph{}
Using the canonical action of $\mathcal{A}$ on $\mathcal{A}_\alpha$, the property \eqref{eq:part_unit_comp_prod} and the action of $\mathcal{Z(A_\alpha)}$ on $\Der_R(\mathcal{A}_\alpha)$ \eqref{zeactiononforms}, one has for any $ X = \sum_\alpha \chi_\alpha(X^\mu_\alpha p^\alpha_\mu) \in \Der_R(\mathcal{A})$ and $z \in \mathcal{Z(A)}$,
\begin{subequations}
\begin{align}
	z \actl \left( \sum_\alpha \chi_\alpha(X^\mu_\alpha p^\alpha_\mu) \right)
	&= \sum_\alpha z \chi_\alpha(X^\mu_\alpha p^\alpha_\mu)
	= \sum_\alpha \chi_\alpha \big( \pi_\alpha(z) \actl (X^\mu_\alpha p^\alpha_\mu) \big)
	\label{eq:der_act_fullA_left} \\
	\big( \sum_\alpha \chi_\alpha(X^\mu_\alpha p^\alpha_\mu ) \big) \actr z
	&= \sum_\alpha \chi_\alpha(X^\mu_\alpha p^\alpha_\mu) z
	= \sum_\alpha \chi_\alpha \big( (X^\mu_\alpha p^\alpha_\mu) \actr \pi_\alpha(z) \big).
	\label{eq:der_act_fullA_right}
\end{align}
    \label{eq:der_act_fullA}
\end{subequations}
This action is the noncommutative analogue of the point-wise product action of the commutative setting. Therefore, one readily deduces
\begin{proposition}
    $\Der_R(\mathcal{A})$ is a $\mathcal{Z(A)}$-module.   
\end{proposition}
\begin{proof}
    One uses \eqref{eq:der_act_fullA_left} and \eqref{eq:der_act_fullA_right}.
\end{proof}

\paragraph{}
One can build the forms on $\mathcal{A}$ and on $\mathcal{A}_\alpha$ as in section \ref{subsec:triv_diff_calc}. Then, the forms of $\mathcal{A}$ can be projected on forms on $\mathcal{A}_\alpha$ and, correspondingly, forms on $\mathcal{A}_\alpha$ can be lifted to forms on $\mathcal{A}$.

\paragraph{}
Explicitly, let $\Omega^n(\mathcal{A})$ be the set of $\mathcal{Z(A)}$-multilinear antisymmetric maps from $\Der(\mathcal{A})^n$ to $\mathcal{A}$, with $\Omega^0(\mathcal{A})= \mathcal{A}$, and $\Omega^n(\mathcal{A}_\alpha)$ be the set of $\mathcal{Z(A_\alpha)}$-multilinear antisymmetric maps from $\Der(\mathcal{A_\alpha})^n$ to $\mathcal{A_\alpha}$, with $\Omega^0(\mathcal{A_\alpha})= \mathcal{A_\alpha}$. Then, one has

\begin{proposition}
    \label{prop:form_loc2glob}
    Let $\rho \in \Omega^n(\mathcal{A})$, then we define for $X_{1\alpha}, \dots, X_{n\alpha} \in \Der(\mathcal{A}_\alpha)$,
    \begin{align}
        \rho^\alpha (X_{1\alpha}, \dots, X_{n\alpha})
        = \pi_\alpha \circ \rho \Big( \chi_\alpha \circ X_{1\alpha} \circ \pi_\alpha, \dots, \chi_\alpha \circ X_{n\alpha} \circ \pi_\alpha \Big)
        \label{eq:form_glob2loc}
    \end{align}
    and one has $\rho^\alpha \in \Omega^n(\mathcal{A}_\alpha)$.

    Correspondingly, let $\rho^\alpha \in \Omega^n(\mathcal{A}_\alpha)$, defining, for $X_1, \dots, X_n \in \Der(\mathcal{A})$,
    \begin{align}
        \rho (X_1, \dots, X_n)
        = \chi_\alpha \circ \rho^\alpha \Big( \pi_\alpha \circ X_1 \circ \chi_\alpha, \dots, \pi_\alpha \circ X_n \circ \chi_\alpha \Big)
        \label{eq:form_loc2glob}
    \end{align}
    and on has $\rho \in \Omega^n(\mathcal{A})$.
\end{proposition}

\begin{proof}
    The proof relies plainly on \eqref{eq:ncpuf_prod} and the fact that $\pi_\alpha$ is a homomorphism. It is standard but we report it for the sake of completeness. The multilinearity is obtained as a composition of linear maps. Only the $\mathcal{Z(A)}$ linearity needs to be checked.

    Let $\rho \in \Omega^n(\mathcal{A})$. We need to prove that $\rho^\alpha$, given by \eqref{eq:form_glob2loc}, is $\mathcal{Z(A_\alpha)}$-multilinear. Let $z \in \mathcal{Z(A)}$, then for any $X_{1\alpha}, \dots, X_{n\alpha} \in \Der(\mathcal{A}_\alpha)$, one has
    \begin{align*}
        \rho^\alpha (X_{1\alpha}, & \dots, X_{j\alpha} \actr \pi_\alpha(z), \dots, X_{n\alpha}) \\
        &= \pi_\alpha \circ \rho \Big( \chi_\alpha \circ X_{1\alpha} \circ \pi_\alpha, \dots, \chi_\alpha \circ (X_{j\alpha} \actr \pi_\alpha(z)) \circ \pi_\alpha,  \dots, \chi_\alpha \circ X_{n\alpha} \circ \pi_\alpha \Big) \\
        &= \pi_\alpha \circ \rho \Big( \chi_\alpha \circ X_{1\alpha} \circ \pi_\alpha, \dots, \chi_\alpha \circ (X_{j\alpha} \pi_\alpha(z)) \circ \pi_\alpha,  \dots, \chi_\alpha \circ X_{n\alpha} \circ \pi_\alpha \Big) \\
        &= \pi_\alpha \circ \rho \Big( \chi_\alpha \circ X_{1\alpha} \circ \pi_\alpha, \dots, (\chi_\alpha \circ X_{j\alpha} \circ \pi_\alpha) z,  \dots, \chi_\alpha \circ X_{n\alpha} \circ \pi_\alpha \Big) \\
        &= \pi_\alpha \left( \rho \Big( \chi_\alpha \circ X_{1\alpha} \circ \pi_\alpha, \dots, \chi_\alpha \circ X_{n\alpha} \circ \pi_\alpha \Big) z \right) \\
        &= \pi_\alpha \circ \rho \Big( \chi_\alpha \circ X_{1\alpha} \circ \pi_\alpha, \dots, \chi_\alpha \circ X_{n\alpha} \circ \pi_\alpha \Big) \pi_\alpha(z) \\
        &= (\rho^\alpha  \actr \pi_\alpha(z)) (X_{1\alpha},\dots, X_{n\alpha})
    \end{align*}
    and similarly for the left action. The proof of $\mathcal{Z(A)}$-multilinearity of $\rho$ in \eqref{eq:form_loc2glob} is similar.
\end{proof}

\paragraph{}
As in section \ref{subsec:triv_diff_calc}, one can define the wedge product through \eqref{eq:triv_diff_calc_wedge} for the forms of $\mathcal{A}$, denoted $\wedge$, and for the forms of $\mathcal{A}_\alpha$, denoted $\wedge_\alpha$. Then one can show that $\wedge_\alpha$ corresponds to the restriction of $\wedge$ to $\mathcal{A}_\alpha$, namely one has
\begin{proposition}
    \label{prop:loct_form_wedge}
    Let $\rho \in \Omega^n(\mathcal{A})$ and $\eta \in \Omega^q(\mathcal{A})$. Define $\rho^\alpha$ and $\eta^\alpha$ through \eqref{eq:form_glob2loc} then one has
    \begin{align}
        \rho^\alpha \wedge_\alpha \eta^\alpha
        = \pi_\alpha (\rho \wedge \eta).
        \label{eq:loct_form_wedge_glob2loc}
    \end{align} 
\end{proposition}

\begin{proof}
    First lets prove \eqref{eq:loct_form_wedge_glob2loc}. Let $X_{1\alpha}, \dots, X_{(n+q)\alpha} \in \Der(\mathcal{A}_\alpha)$, one computes
    \begin{align*}
        (\rho^\alpha \wedge_\alpha \eta^\alpha) & (X_{1\alpha}, \dots X_{(n+q)\alpha}) \\
	    &= \frac{1}{n!q!} \sum_{\sigma \in \mathfrak{S}_{n+q}} (-1)^{\sign(\sigma)} \rho^\alpha(X_{\sigma(1)\alpha}, \dots, X_{\sigma(n)\alpha}) \eta^\alpha(X_{\sigma(n+1)\alpha}, \dots, X_{\sigma(n+q)\alpha}) \\
	    &= \frac{1}{n!q!} \sum_{\sigma \in \mathfrak{S}_{n+q}} (-1)^{\sign(\sigma)} \left( \pi_\alpha \circ \rho \Big( X_{\sigma(1)}, \dots, X_{\sigma(n)} \Big) \right) \left( \pi_\alpha \circ \eta \Big( X_{\sigma(n+1)}, \dots, X_{\sigma(n+q)} \Big) \right) \\
	    &= \pi_\alpha \left( \frac{1}{n!q!} \sum_{\sigma \in \mathfrak{S}_{n+q}} (-1)^{\sign(\sigma)} \rho \Big( X_{\sigma(1)}, \dots, X_{\sigma(n)} \Big) \eta \Big( X_{\sigma(n+1)}, \dots, X_{\sigma(n+q)} \Big) \right) \\
        &= \pi_\alpha \left( (\rho \wedge \eta) \Big( X_{\sigma(1)}, \dots, X_{\sigma(n+q)} \Big) \right)
    \end{align*}
    where we denoted $X_j = \chi_\alpha \circ X_{j\alpha} \circ \pi_\alpha$ and we used that $\pi_\alpha$ is a homomorphism.
\end{proof}

\paragraph{}
In this context, the differential can be defined through the Koszul formula \eqref{eq:triv_diff_calc_diff} on $\mathcal{A}$, denoted by $\dd$, and on $\mathcal{A}_\alpha$, denoted $\dd_\alpha$. These two should match on $\mathcal{A}_\alpha$ since the differential is local. In the commutative case, it translates to the fact that the restriction of $\dd$ to $\mathcal{A}_\alpha$ equals $\dd_\alpha$ (see for example \cite{Rudolph_2012} Remark 4.1.5).

\paragraph{}
However, it turns out that we did not define a notion of restriction for derivations and forms yet. Indeed, $X_\alpha$ or $\rho^\alpha$ defined above do not correspond to the restriction of $X$ and $\rho$ (respectively) to $\mathcal{A}_\alpha$ since their expression involves $\chi_\alpha$. To define such a restriction, one rather use a {\it{right inverse}} of $\pi_\alpha$, denoted ${}^R\pi_\alpha^{-1}$, which exists because $\pi_\alpha$ is surjective.

Thus, we define for any function on $\mathcal{A}$, say $X$, its restriction to $\mathcal{A}_\alpha$, by
\begin{equation}
    X|_\alpha = \pi_\alpha \circ X \circ {}^R\pi_\alpha^{-1}.    
\end{equation}
The restricted form is then defined as \eqref{eq:form_glob2loc}, replacing $\chi_\alpha$ by ${}^R\pi_\alpha^{-1}$.
\begin{proposition}
    Let $\rho \in \Omega^n(\mathcal{A})$, then one has
    \begin{align}
        (\dd \rho)|_\alpha = \dd_\alpha \rho|_\alpha
        \label{eq:loct_diff_loc}
    \end{align}
\end{proposition}

\begin{proof}
    Let $X_{1\alpha}, \dots, X_{(n+1)\alpha} \in \Der(\mathcal{A}_\alpha)$, one has
    \begin{align*}
        (\dd \rho)|_\alpha & (X_{1\alpha}, \dots, X_{(n+1)\alpha}) \\
        &= \pi_\alpha \Bigg( \sum_{j=1}^{n+1} (-1)^{j+1} {}^R\pi_\alpha^{-1} \circ X_{j\alpha} \circ \pi_\alpha \circ \rho(X_{1}, \dots, \omitel{j}, \dots, X_{n+1}) \\
		&+ 
		\sum_{1\leqslant j < k\leqslant n+1} (-1)^{j+k} \rho \Big([X_{j},X_{k}], X_{1}, \dots, \omitel{j}, \dots, \omitel{k}, \dots, X_{n+1} \Big) \Bigg). \\
        &= \sum_{j=1}^{n+1} (-1)^{j+1} X_{j\alpha} \circ \pi_\alpha \circ \rho(X_{1}, \dots, \omitel{j}, \dots, X_{n+1}) \\
		&+ 
		\sum_{1\leqslant j < k\leqslant n+1} (-1)^{j+k} \pi_\alpha \circ \rho \Big([X_{j},X_{k}], X_{1}, \dots, \omitel{j}, \dots, \omitel{k}, \dots, X_{n+1} \Big). \\
        &= \dd_\alpha \rho|_\alpha (X_{1\alpha}, \dots, X_{(n+1)\alpha})
    \end{align*}
    where we set $X_j = {}^R\pi_\alpha^{-1} \circ X_{j\alpha} \circ \pi_\alpha$. Note that this proof uses 
    \begin{equation}
        [{}^R\pi_\alpha^{-1} \circ X_{j\alpha} \circ \pi_\alpha, {}^R\pi_\alpha^{-1} \circ X_{k\alpha} \circ \pi_\alpha] = {}^R\pi_\alpha^{-1} \circ [X_{j\alpha},X_{k\alpha}] \circ \pi_\alpha.
    \end{equation}
\end{proof}

\paragraph{}
Note that one can also formulate Proposition \ref{prop:der_loc2glob} and \ref{prop:form_loc2glob} by replacing $\chi_\alpha$ by ${}^R\pi_\alpha^{-1}$.

\paragraph{}
Now, we turn to the locally trivial setting we introduced. From section \ref{subsec:triv_diff_calc} and Definition \ref{def:loct_tan_bun}, as each local algebra $\mathcal{A}_\alpha$ is trivial, one has $\Omega^1_R(\mathcal{A}_\alpha) = \mathcal{A}_\alpha \otimes \mathfrak{D}_\kappa'$. Moreover, in view of \eqref{eq:loct_fullderR_def} and \eqref{eq:form_loc2glob}, one can define 
\begin{align}
    \Omega_R^1(\mathcal{A}) 
    = \left\{ \sum_\alpha \chi_\alpha \circ \rho^\alpha ( \pi_\alpha \circ \cdot \circ \chi_\alpha), \ \rho^\alpha \in \Omega_R^1(\mathcal{A}_\alpha) \right\}.
    \label{eq:loct_fullformR_def}
\end{align}
Note that $\pi_\alpha \circ \cdot \circ \chi_\alpha : \Der(\mathcal{A}) \to \Der(\mathcal{A}_\alpha)$ must be added for the same reasons as in \eqref{eq:form_loc2glob}, that is that $\rho$ takes argument in $\Der(\mathcal{A})$ and $\rho^\alpha$ takes argument in $\Der(\mathcal{A}_\alpha)$.

\begin{proposition}
    \label{prop:loct_form_loc2glob}
    Any $\rho \in \Omega^1_R(\mathcal{A})$ can be expressed as
    \begin{align}
        \rho = \sum_\alpha \chi_\alpha (\rho^\alpha_\mu \mathfrak{X}^\mu_\alpha)
        \label{eq:loct_fullformR_loc2glob}
    \end{align}
    where $\rho^\alpha_\mu \in \mathcal{A}_\alpha$ and $\mathfrak{X}^\mu_\alpha = \mathfrak{X}^\mu(\pi_\alpha \circ \cdot \circ \chi_\alpha)$, with $\mathfrak{X}^\mu$ a basis element of $\mathcal{T}_\kappa$.
\end{proposition}

\begin{proof}
    Since $\Omega_R^1(\mathcal{A}_\alpha) = \mathcal{A}_\alpha \otimes \mathfrak{D}_\kappa'$, then for any $\rho^\alpha \in \Omega_R^1(\mathcal{A}_\alpha)$, one can find $\rho^\alpha_\mu \in \mathcal{A}_\alpha$ such that $\rho^\alpha = \rho^\alpha_\mu \mathfrak{X}^\mu$, with $\{ \mathfrak{X}^\mu \}_\mu$ being the generators of $\mathcal{T}_\kappa$. Then, for any $X \in \Der(\mathcal{A})$,
    \begin{align*}
        \rho(X)
        = \sum_\alpha \chi_\alpha \circ \rho^\alpha ( \pi_\alpha \circ X \circ \chi_\alpha)
        = \sum_\alpha \chi_\alpha \Big( \rho^\alpha_\mu \mathfrak{X}^\mu( \pi_\alpha \circ X \circ \chi_\alpha) \Big)
        = \sum_\alpha \chi_\alpha ( \rho^\alpha_\mu \mathfrak{X}^\mu_\alpha(X) )
    \end{align*}
\end{proof}

\paragraph{}
For consistency we can check that $\Omega^1_R(\mathcal{A})$, defined as such, is a $\mathcal{A}$-module by the extension of the $\mathcal{A}_\alpha$ action \eqref{eq:triv_form_act} on $\Omega^1_R(\mathcal{A}_\alpha)$. Explicitly, for any $\rho = \sum_\alpha \chi_\alpha(\rho^\alpha_\mu \mathfrak{X}^\mu_\alpha) \in \Omega_R^1(\mathcal{A})$ and $f \in \mathcal{A}$,
\begin{subequations}
\begin{align}
	f \actl \big( \sum_\alpha \chi_\alpha(\rho^\alpha_\mu \mathfrak{X}^\mu_\alpha) \big)
	&= \sum_\alpha f \chi_\alpha(\rho^\alpha_\mu \mathfrak{X}^\mu_\alpha)
	= \sum_\alpha \chi_\alpha \big( f \actl (\rho^\alpha_\mu \mathfrak{X}^\mu_\alpha) \big),
	\label{eq:form_act_fullA_left}\\
	\big( \sum_\alpha \chi_\alpha(\rho^\alpha_\mu \mathfrak{X}^\mu_\alpha) \big) \actr f
	&= \sum_\alpha \chi_\alpha(\rho^\alpha_\mu f \mathfrak{X}^\mu_\alpha)
	= \sum_\alpha \chi_\alpha \big( (\rho^\alpha_\mu \mathfrak{X}^\mu_\alpha) \actr f \big).
	\label{eq:form_act_fullA_right}
\end{align}
\label{eq:form_act_fullA}
\end{subequations}

As in the previous section, this action is the noncommutative analogue of the point-wise product action of the commutative setting. Therefore,
\begin{proposition}
    $\Omega_R^1(\mathcal{A})$ is a $\mathcal{A}$-module.   
\end{proposition}
\begin{proof}
    One uses \eqref{eq:form_act_fullA_left} and \eqref{eq:form_act_fullA_right}.
\end{proof}

\subsection{Global quantisation and integration}
\label{subsec:loc2glob_quant}
\paragraph{}
The role of the partition of unity to link locally defined and globally defined objects does not only apply to tensorial objects. In this section, we use the partition of unity to link the local quantisation procedure we defined in \eqref{eq:qst_loc_sp_rel} to a global one (see \eqref{eq:qst_glob_sp}). Furthermore, we also want to build a notion of integral on the full $\mathcal{A}$ that can be fully determined by an integral on $\kM$. This is done in two step: first we define an integral on $\mathcal{A}_\alpha$ from an integral on $\kM$ through the coordinate charts, and then we glue these integrals thanks to a partition of unity, similarly as in the commutative case. Beyond these definitions, we prove that they do not depend on both the choice of the covering of algebra and second on the choice of the partition of unity.

In this subsection only, we explicitly denote the product of $\mathcal{A}$ by $\star$ and the one of $\mathcal{A}_\alpha$ by $\star_\alpha$. These notations are here to underline the link between the local ($\star_\alpha$) and global ($\star$) quantisation.

\paragraph{}
If one pays a close attention to the expression \ref{it:ncpu_id} of Definition \ref{def:nc_part_unit}, written under the form \eqref{eq:ncpu_loc2glob}, together with the fact that $\pi_\alpha$ is by definition a homomorphism (\textit{i.e.}\ $\pi_\alpha( f \star g) = \pi_\alpha(f) \star_\alpha \pi_\alpha(g)$, then the following holds
\begin{align}
    f \star g
    = \sum_\alpha \chi_\alpha \big( \pi_\alpha(f) \star_\alpha \pi_\alpha(g) \big),
    \label{eq:qst_glob_sp}
\end{align}
for any $f,g \in \mathcal{A}$. It implies that the quantisation of the full noncommutative spacetime is actually directly inherited from the quantisation of the local noncommutative spacetimes and so inherited from $\kM$.

\paragraph{}
It is important to underline that the expression \eqref{eq:qst_glob_sp} does not depend on the choice of the covering of algebra or the choice of the partition of unity. Indeed, if one consider another covering $\{(\mathcal{A}_{\tilde{\beta}}, \star_{\tilde{\beta}})\}_{\tilde{\beta}}$, with an adapted partition of unity $\{\tilde{\chi}_{\tilde{\beta}}\}_{\tilde{\beta}}$, then
\begin{align*}
    f \star g
    &= \sum_\alpha \chi_\alpha \big( \pi_\alpha(f) \star_\alpha \pi_\alpha(g) \big)
    = \sum_{\alpha, \tilde{\beta}} \tilde{\chi}_{\tilde{\beta}} \circ \pi_{\tilde{\beta}} \circ \chi_\alpha \circ \pi_\alpha \big( f \star g \big) \\
    &= \sum_{\alpha, \tilde{\beta}} \chi_{\alpha\tilde{\beta}} \circ \pi_{\alpha\tilde{\beta}} ( f \star g )
    = \sum_{\tilde{\beta}} \tilde{\chi}_{\tilde{\beta}} \big( \pi_{\tilde{\beta}}(f) \star_{\tilde{\beta}} \pi_{\tilde{\beta}}(g) \big).
\end{align*}
This computation only uses the relation \eqref{eq:ncpuf_prod}, with $\chi_{\alpha\tilde{\beta}} = \chi_\alpha \star \chi_{\tilde{\beta}}$.

\paragraph{}
The involution of the spacetime algebra $\mathcal{A}$ can also be derived from the local involutions in a similar fashion. Explicitly, one has
\begin{align}
    f^* = \sum_\alpha \chi_\alpha \big( \pi_\alpha(f)^* \big)
    \label{eq:qst_glob_inv}
\end{align}
for any $f \in \mathcal{A}$. This expression is one again independent of the choice of covering and partition of unity.

\paragraph{}
We here define the notion of integration on a local algebra having coordinate charts sending to $\kappa$-Minkowski, as defined in \eqref{eq:qst_qpb}. Similarly to the commutative setting, one considers, for any $\omega_\alpha \in \Omega^{d+1}(\mathcal{A}_\alpha)$,
\begin{align}
    \int_{\mathcal{A}_\alpha} \omega_\alpha
    = \int_{\varphi^\alpha(\mathcal{A}_\alpha) \subset \kM} \varphi^\alpha( \omega ).
\end{align}
The integration domain is now a subset of $\kM$, therefore the integral on the algebra is fully determined by the integral on $\kappa$-Minkowski.

Correspondingly, the integration on the global algebra is defined through the partition of unity, for any $\omega \in \Omega^{d+1}(\mathcal{A})$,
\begin{align}
\begin{aligned}
    \int_{\mathcal{A}} \omega
    &= \sum_\alpha \int_{\mathcal{A}_\alpha} \pi_\alpha( \chi_\alpha \star \omega)
    = \sum_\alpha \int_{\varphi^\alpha(\mathcal{A}_\alpha)} \varphi^\alpha \circ \pi_\alpha(\chi_\alpha \star \omega) \\
    &= \sum_\alpha \int_{\varphi^\alpha(\mathcal{A}_\alpha)} \varphi^\alpha \circ \pi_\alpha(\chi_\alpha) \star_\kappa  \varphi^\alpha \circ \pi_\alpha(\omega).
\end{aligned}
    \label{eq:qst_gint_def}
\end{align}
One should note that the previous definitions do not depend on the choice of the covering of algebra or of the partition of unity.

\begin{proof}
    Consider two covering of algebras $\{\mathcal{A}_\alpha\}_\alpha$ and $\{\mathcal{A}_{\tilde{\beta}}\}_{\tilde{\beta}}$ with adapted partitions of unity $\{\chi_\alpha\}_\alpha$ and $\{\tilde{\chi}_{\tilde{\beta}}\}_{\tilde{\beta}}$ respectively. One can split the whole integral over $\mathcal{A}$ as integrals over the $\mathcal{A}_\alpha$'s through \eqref{eq:qst_gint_def}, and then split each integral over $\mathcal{A}_\alpha$ as integrals over $\mathcal{A}_{\alpha \tilde{\beta}}$ using again \eqref{eq:qst_gint_def}. It writes
    \begin{align}
    \begin{aligned}
        \int_\mathcal{A} \omega
        &= \sum_\alpha \int_{\mathcal{A}_\alpha} \pi_\alpha(\chi_\alpha \star \omega)
        = \sum_{\alpha, \tilde{\beta}} \int_{\mathcal{A}_{\alpha\tilde{\beta}}} \pi^\alpha_{\tilde{\beta}} \big( \tilde{\chi}_{\tilde{\beta}} \star_\alpha \pi_\alpha( \chi_\alpha \star \omega) \big) \\
        &= \sum_{\alpha, \tilde{\beta}} \int_{\mathcal{A}_{\alpha\tilde{\beta}}} \pi^\alpha_{\tilde{\beta}} \circ \pi_\alpha( \tilde{\chi}_{\tilde{\beta}} \star \chi_\alpha \star \omega)
        = \sum_{\alpha, \tilde{\beta}} \int_{\mathcal{A}_{\alpha\tilde{\beta}}} \pi_{\alpha \tilde{\beta}} ( \chi_{\alpha \tilde{\beta}} \star \omega)
    \end{aligned}
    \end{align}
    and by reverting the roles of $\alpha$ and $\tilde{\beta}$, one obtains
    \begin{align}
        \int_\mathcal{A} \omega
        = \sum_\alpha \int_{\mathcal{A}_\alpha} \pi_\alpha ( \chi_\alpha \star \omega)
        = \sum_{\tilde{\beta}} \int_{\mathcal{A}_{\tilde{\beta}}} \pi_{\tilde{\beta}} ( \chi_{\tilde{\beta}} \star \omega ).
    \end{align}
\end{proof}

\section{Conclusion}
\label{sec:conc}
\paragraph{}
In this paper, we constructed a noncommutative version of a locally trivial tangent bundle together with a noncommutative partition of unity. The starting idea is to request that the (noncommutative analogue of) local tangent space is a $\kappa$-Minkowski spacetime, which thus mimics the commutative situation. Here, the notion of a ``local'' algebra is built as a quotient of the total algebra by an ideal, in order to mimic the commutative case. Furthermore, if one consider that the local algebra correspond to a star-product quantisation of a local open set of the full spacetime manifold, then one can obtain a natural star-product from the star-product quantisation of the Minkowski space (here taken to be $\kappa$-Minkowski). This is done through the generalisation of local coordinate charts to the noncommutative case. From this construction, one can straightforwardly build a (co)module algebra structure on the local algebra, which is at the very basis of how $\kappa$-Minkowski ends up being the noncommutative local tangent space.

The noncommutative partition of unity, defined in this paper, was developed to complete this scheme. Explicitly, it allows to express global derivations and forms as sums of their local ``components''. As examples, we restricted the space of local derivations to the linear span of the four generators of $\kappa$-Minkowski and defined the global derivations as the gluing of the local ones. This restricted global derivations are shown to be a $\mathcal{Z(A)}$-submodule of $\Der(\mathcal{A})$. Furthermore, the restricted local one-forms are built as the dual space of the restricted local derivations, in the algebraic sense. Once again, the restricted global one-forms are defined as the gluing of the local ones through the partition of unity and are shown to be a $\mathcal{A}$-submodule of the global one-forms. Finally, we showed that the global star-product and a coherent definition of integral on the global quantum spacetime are obtained from their local counterparts with the use of the partition, but that these global objects does not depend the choice of either the partition or the covering of algebra.

\paragraph{}
The above discussion points toward several improvements or generalisations. The following list is far from being exhaustive but is thought to be relevant for physics matters. 

\paragraph{}
First, the discussion in this paper is purely algebraic in nature. The next interesting step would be to provide an action generalising the Einstein-Hilbert action. We will come back to some the notion of action in this context in a forthcoming publication \cite{tbp24}. However, such an action can only be formulated thanks to a topological setting, which is absent \textit{a priori} here. To do so $\mathcal{A}$ and the $\mathcal{A}_\alpha$'s must be turned into topological algebras. A topology may imply relations to make consistent topologies of the global algebra and the local ones. A way of including a consistent topology would be to use sheaf theory.

Indeed, gauge theory on principal fiber bundles has a quantum formulation within sheaf theory \cite{Pflaum_1994} and \cite{Aschieri_2021b, Aschieri_2021c}. The basic object of such formulations is a quantum ringed space involving a topological space (for us it would be $\mathcal{A}$) and a sheaf of noncommutative algebra over it (the stalks of the sheaf would be our $\mathcal{A}_\alpha$'s). The structure group of the quantum principal fiber bundle is a Hopf algebra $\mathcal{H}$ and the principal bundle is a sheaf of $\mathcal{H}$-comodule algebras for which all invariant stalks correspond to the base stalks. 

In our construction, we require that every stalks $\mathcal{A}_\alpha$ are $\mathcal{T}_\kappa$-comodule algebra. Then a straightforward generalisation of our formalism in terms of sheaf theory would be to impose $\mathcal{H=T_\kappa}$. Yet, in that case, the full structure of the bundle remain to be determined, like the invariant elements of the coaction.

Finally, the notion of a partition of unity is well-defined in fine sheaves (see for example \cite{Warner_1983} section 5.10). As it is a cornerstone of our construction, it may give our setting more mathematical coherence.

\appendix


\section{The \tops{$\kappa$}{kappa}-Minkowski space}
\label{ap:kM}
\paragraph{}
In this section, we briefly present the $\kappa$-Minkowski space \cite{Lukierski_2017}. This quantum space correspond to the deformation of the Minkowski space and is thought to feature some quantum gravity effects. For a recent review on quantum gravity phenomenology see \cite{Addazi_2022}.

\paragraph{}
It turns out that the $\kappa$-Minkowski space is rigidly linked to a deformation of the Poincar\'{e} algebra \cite{Lukierski_1991}, called the $\kappa$-Poincar\'e algebra, which can be viewed as the quantum (\textit{i.e.}\ noncommutative) analogue to its algebra of symmetries.

The $\kappa$-Poincar\'{e} algebra $\mathcal{P}_\kappa$ may be defined as the bicrossproduct \cite{MR1994} $\mathcal{P}_\kappa = \mathcal{T}_\kappa \bicros U\mathfrak{so}(1,d)$, where $\mathcal{T}_\kappa$ denotes the set of deformed translations, generated by $\{P_\mu\}_{\mu \in \{0, \dots, d\}}$, and $U\mathfrak{so}(1,d)$ is the universal enveloping algebra of the rotations and boosts, generated by $\{M_j\}_{j \in \{1, \dots, d\}}$ and $\{N_j\}_{j \in \{1, \dots, d\}}$. We consider the so-called Majid-Ruegg basis \cite{MR1994}, thus trading $P_0$ for $\mathcal{E} = e^{-\frac{P_0}{\kappa}}$. The Hopf algebra structure of $\mathcal{P}_\kappa$ is as follows
\begin{subequations}
\begin{align}
		[M_j,M_k] &= i \tensor{\epsilon}{_{jk}^l} M_l, & 
		[M_j,N_k] &= i\tensor{\epsilon}{_{jk}^l} N_l, & 
		[N_j,N_k] &= -i\tensor{\epsilon}{_{jk}^l} M_l, \\
		[M_j,P_k] &= i\tensor{\epsilon}{_{jk}^l} P_l, &
		[P_j,\mathcal{E}] &= [M_j,\mathcal{E}]=0, &
		[P_j, P_k] &= 0,
		\label{eq:kM_kP_Hopf_alg_alg_tran}
\end{align}%
    \vspace{\dimexpr-\abovedisplayskip-\baselineskip+\jot}%
\begin{align}
    [N_j,\mathcal{E}] &= -\dfrac{i}{\kappa}P_j\mathcal{E}, &
    [N_j,P_k] = - \frac{i}{2} \delta_{jk} \left( \kappa(1-\mathcal{E}^2) + \frac{1}{\kappa} P_l P^l \right) + \frac{i}{\kappa} P_j P_k,
\end{align}%
\begin{align}
    \Delta P_0 &= P_0 \otimes 1 + 1 \otimes P_0, &
	\Delta P_j &= P_j \otimes 1 + \mathcal{E} \otimes P_j, 
	\label{eq:kM_kP_Hopf_alg_coalg_tran} \\
	\Delta \mathcal{E} &= \mathcal{E} \otimes \mathcal{E}, &
	\Delta M_j &= M_j \otimes 1 + 1 \otimes M_j,
\end{align}%
    \vspace{\dimexpr-\abovedisplayskip-\baselineskip+\jot}%
\begin{align}
    \Delta N_j = N_j\otimes 1 + \mathcal{E}\otimes N_j - \frac{1}{\kappa} \tensor{\epsilon}{_j^{kl}} P_k \otimes M_l,
\end{align}%
\begin{align}
    \varepsilon(P_0) = \varepsilon (P_j) = \varepsilon(M_j) = \varepsilon(N_j) = 0, &&
    \varepsilon(\mathcal{E})=1,
\end{align}%
\begin{align}
		S(P_0)&=-P_0, &
		S(\mathcal{E}) &= \mathcal{E}^{-1}, &
		S(P_j) &= -\mathcal{E}^{-1}P_j,
\end{align}%
    \vspace{\dimexpr-\abovedisplayskip-\baselineskip+\jot}%
\begin{align}
		S(M_j) &= -M_j, &
		S(N_j) &= -\mathcal{E}^{-1}(N_j-\dfrac{1}{\kappa} \tensor{\epsilon}{_j^{kl}} P_k M_l).
\end{align}%
    \label{eq:kM_kP_Hopf_alg}
\end{subequations}
The deformation parameter $\kappa$ has mass dimension $1$. It is often associated to the Planck mass.

\paragraph{}
One then defines $\kappa$-Minkowski as the dual Hopf algebra of the deformed translations, that is $\kM = \mathcal{T}_\kappa'$. This duality allows to compute the full Hopf algebra structure of $\kM$. One obtains that $\kM$ is generated
by $\{p_\mu\}_{\mu \in \{0, \dots, d\}}$ that satisfies
\begin{align}
    [p_0, p_j] &= \frac{i}{\kappa} p_j, &
    [p_j, p_k] &= 0, 
    \tag{\ref{eq:kM_kM_Hopf_alg_alg}}\\
    \Delta(p_\mu) &= p_\mu \otimes 1 + 1 \otimes p_\mu, &
    S(p_\mu) &= - p_\mu,
    \tag{\ref{eq:kM_kM_Hopf_alg_coalg}}
\end{align}
where we denoted, $[p, q] = p \star_\kappa q - q \star_\kappa p$ and $\star_\kappa$ is the associative star-product of $\kM$.

\paragraph{}
In terms of symmetries, the dual structure allows $\mathcal{P}_\kappa$ to act on $\kM$ through
\begin{subequations}
\begin{align}
    (P^\mu \actl f)(p) &= -i \frac{\partial}{\partial p_\mu} f(p), &
    (\mathcal{E} \actl f)(p) &= f\left(p_0 + \frac{i}{\kappa}, \vec{p} \right), \\
    (M^j \actl f)(x) &= \big( \tensor{\epsilon}{^{jk}_{l}} p_k P^l \actl f \big) (p), &&
\end{align}%
    \vspace{\dimexpr-\abovedisplayskip-\baselineskip+\jot}%
\begin{align}
    (N^j \actl f)(p)
    &= \left( \Big( \frac{1}{2}p^j \big( \kappa (1 - \mathcal{E}^2) + \frac{1}{\kappa} P_lP^l \big) + p_0 P^j - \frac{i}{\kappa} p_k P^k P^j  \Big) \actl f \right)(x),
\end{align}
    \label{eq:kM_kP_action}
\end{subequations}
where $f \in \kM$.

\paragraph{}
As \eqref{eq:kM_kM_Hopf_alg_alg} correspond to the Lie algebra of the affine group $\mathbb{R} \ltimes \mathbb{R}^d$, one can use the convolution algebra machinery to built the following star-product \cite{DS, PW2018}
\begin{align}
    (f \star_\kappa g)(p) &= \int \frac{\dd y^0}{2 \pi} \dd q_0\ e^{-i y^0 q_0} f(p_0 + q_0, \vec{p}) g(p_0, e^{-y^0/\kappa} \vec{p}), 
    \label{eq:kM_star-kappa}\\
    f^\dag(p) &= \int \frac{\dd y^0}{2\pi} \dd q_0\ e^{-i y^0 q_0} \overline{f} (p_0 + q_0, e^{-y^0/\kappa} \vec{p})
    \label{eq:kM_invol-kappa}.
\end{align}

\section*{Acknowledgement}
\paragraph{}
The authors thank the Action CA21109 CaLISTA ``Cartan geometry, Lie, Integrable Systems, quantum group Theories for Applications'', from the European Cooperation in Science and Technology (COST).

\paragraph{}
K.H.~has been partially supported, during this work, by the grants CNS2023-143760, funded by NextGenerationEU and PID2023-148373NB-I00 funded by \\
MCIN/AEI/10.13039/501100011033/FEDER – UE.

\section*{Declaration of competing interests}
\paragraph{}
The authors declare that they have no known competing financial interests or non-financial interest or personal relationships that could have appeared to influence the work reported in this paper.


\end{document}